\documentclass[12pt]{article}

\usepackage[margin=2.5cm]{geometry}

\usepackage{amsmath,amsthm,amsfonts,amscd,amssymb,bbm,mathrsfs, enumerate}
\usepackage{pgf,graphicx,tikz}
\usepackage[pdfborder={0 0 0}]{hyperref}
\usetikzlibrary{arrows}

\def\RR{{\mathbb R}}
\def\CC{{\mathbb C}}
\def\NN{{\mathbb N}}
\def\ZZ{{\mathbb Z}}

\def\SS{{\mathbb S}}

\def\A{{\mathcal A}}

\def\H{{\mathcal H}}

\def\M{{\mathcal M}}

\def\a{\alpha}
\def\b{\beta}

\def\e{\varepsilon}
\def\f{\varphi}
\def\g{\gamma}
\def\G{\Gamma}

\def\k{\kappa}

\def\l{\lambda}

\def\t{\tau}
\def\th{\theta}

\def\x{\xi}

\def\z{\zeta}

\def\1{{\mathbbm 1}}

\def\u1net{{\A^{(0)}}}

\def\diff{{\rm Diff}}

\def\diffs1{\diff(S^1)}
\def\mob{{\rm M\ddot{o}b}}
\def\mob2{{\rm M\ddot{o}b}^{(2)}}

\def\psl2r{{\rm PSL}(2,\RR)}
\def\sl2r{{\rm SL}(2,\RR)}
\def\su11{{\rm SU}(1,1)}
\def\2dmob{{\overline{\psl2r}\times\overline{\psl2r}}}
\def\<{\langle}
\def\>{\rangle}

\newcommand{\tout}{\mathrm{out}}
\newcommand{\tin}{\mathrm{in}}

\newcommand{\shift}{\Delta^{\frac12}}

\newcommand{\strip}{\RR + i(-\pi,0)}
\newcommand{\hardy}{H^2(\SS_{-\pi,0})}
\newcommand{\blaschke}{\mathrm{Bl}}
\newcommand{\mfbar}{M_{\overline f}}

\newcommand{\dom}{\mathrm{Dom}} 
\newcommand{\im}{\mathrm{Im}\,} 
\newcommand{\re}{\mathrm{Re}\,}

\newtheorem{theorem}{Theorem}[section]

\newtheorem{corollary}[theorem]{Corollary}
\newtheorem{proposition}[theorem]{Proposition}
\newtheorem{lemma}[theorem]{Lemma}
\theoremstyle{remark}

\newtheorem{example}[theorem]{Example}

\title{Self-adjointness of bound state operators in integrable quantum field theory}
\date{}
\author{
{\bf Yoh Tanimoto}\\
e-mail: {\tt hoyt@ms.u-tokyo.ac.jp}\\
Graduate School of Mathematical Sciences, The University of Tokyo\\
and Institut f\"ur Theoretische Physik, G\"ottingen University\\
3-8-1 Komaba Meguro-ku Tokyo 153-8914, Japan.\\
JSPS SPD postdoctoral fellow\\
}
\begin{document}
\maketitle
\begin{abstract}
 We study self-adjoint extensions of operators which are the product of
 the multiplication operator by an analytic function and the analytic continuation in a strip.
 We compute the deficiency indices of the product operator for a wide class of analytic
 functions.
 For functions of a particular form, we point out the existence of
 a self-adjoint extension which is unitarily equivalent to the analytic-continuation
 operation.

 They appear in integrable quantum field theories as the one-particle component of the operators which
 realize the bound states of elementary particles and the existence of self-adjoint extension
 is a necessary step for the construction of Haag-Kastler net for such models.
\end{abstract}

\section{Introduction}\label{introduction}

Products of unbounded operators are subtle.
If $A$ is a densely defined closed unbounded operator and $x$ is a bounded operator on a Hilbert space,
then it is easy to see that $Ax$ on the obvious domain $\{\xi: x\xi \in \dom(A)\}$
is closed but it may fail to be densely defined. On the other hand,
$xA$ is densely defined, but it is not necessarily closed (in general,
$xA$ is closable if and only if $A^*x^* = (xA)^*$ is densely defined \cite[Theorem 13.2]{Rudin91}.
This is not the case if one considers $A = A^*$ unbounded and take $x$ to be
the projection onto a subspace whose vectors are not in $\dom(A)$).

In \cite{CT15-1}, we encountered operators of the form $\mfbar\shift$ (up to a rescaling), where
$\mfbar$ is the multiplication operator by an analytic function $\overline f$
and $\shift$ is the analytic continuation:
\[
 (\mfbar\shift\xi)(\theta) = \overline{f(\theta)}\xi(\theta - \pi i).
\]
We give the precise definitions in Section \ref{preliminaries}.
Under the condition that $\overline{f(\theta)} = f(\theta - \pi i)$,
this operator is symmetric. Then the natural question arises whether
this operator is (essentially) self-adjoint or not, and if not,
what its self-adjoint extensions are.
It turns out that this question highly depends on $f$: its zeros and the decay rate at $\theta \to \pm\infty$.

This is not only a purely mathematical problem. The operator $\mfbar\shift$
appears as a building block of a quantum observable in certain two-dimensional quantum field
theories \cite{CT15-1}. In a relativistic quantum field theory, it is required that
observables localized in spacelike-separated regions should strongly commute \cite{Haag96}.
Therefore, it is an important problem to classify the self-adjoint extensions
of $\mfbar\shift$ and choose a right one.

Although the operator $\mfbar\shift$ looks simple and its extension theory has interesting
features as we will see, it has apparently been treated neither in textbooks, e.g.\!
\cite{vonNeumann55, RS55, Goldberg66, Helmberg69, Friedrichs73, Kato76, RSII, Yosida80, Weidmann80, DS88,
Rudin91, AG93, BSU96, Teschl09, Schmuedgen12}
nor in a recent review of self-adjointness in quantum physics \cite{IP15}.
In this paper, we go back to the basics, namely we compute deficiency indices of the operator.
We will see that the deficiency indices depend highly on the choice of $f$,
furthermore, there appears to be no canonical choice of a self-adjoint extension for a generic $f$.
Therefore, we restrict ourselves to a certain subclass of functions.
For such functions, we can find a self-adjoint extension of $\mfbar\shift$ which is
unitarily equivalent to $\shift$. This particular extension will be useful
in the original context of integrable quantum field theory and bound states \cite{Tanimoto15-2}.

For a generic $f$, we take its Beurling factorization and we reduce the computation of
deficiency indices to each factors. We find especially, when $f$ has zeros in the strip, that $\mfbar\shift$
may have different deficiency indices, and may have no self-adjoint extension.
By considering our applications \cite{CT15-1}, this forces us to pick the subclass of functions which are
a square of another function $h$: $f = h^2$.
If $f$ such a square, there is a canonical choice of a self-adjoint extension
which is unitarily equivalent to $\shift$.

This paper is organized as follows. In Section \ref{preliminaries}, we state precisely the assumptions
and the problem. We compute in Sections \ref{di}, \ref{summary} the deficiency indices of $\mfbar\shift$ and
obtain an explicit form for the vectors in the deficiency subspaces for certain functions $f$. We also get an expression
of the polar decomposition.
In Section \ref{squares}, we consider functions $f$ which is a square and construct
a canonical self-adjoint extension of $\mfbar\shift$.

\section{Preliminaries}\label{preliminaries}
\subsection{Hardy spaces}\label{hardy}
We denote by $H^2(\SS_{a,b})$ the space of analytic functions $\xi$ in the strip
$\SS_{a,b} := \RR + i(a,b)$, $a<b$, such that
$\xi(\theta + i\l)$ is in $L^2(\RR)$ (with the Lebesgue measure on $\RR$, which we will omit in the rest)
for a fixed $\l$ and the norms $\|\xi(\,\cdot +i\l)\|$ are uniformly bounded
for $\l \in (a,b)$. They are called the Hardy space based on the strip $\SS_{a,b}$.
Let us fix one such Hardy space $\hardy$.
For each element $\xi \in \hardy$, the limits $\lim_{\epsilon\to 0}\xi(\,\cdot - i\epsilon)$ and
$\lim_{\epsilon\to 0}\xi(\,\cdot - (\pi - \epsilon)i)$
exist in the sense of $L^2(\RR)$ \cite[Corollary III.2.10]{SW71}. Let us denote these boundary values by
$\xi(\theta), \xi(\theta - \pi i)$ for simplicity. Then,
for the Fourier transform $\hat \xi(t)$ of $\xi(\theta)$, $e^{\pi t}\hat\xi(t)$ is $L^2$
and it holds that
\[
 \xi(\theta + i\l) = \frac1{\sqrt{2\pi}}\int dt\,e^{i(\theta+i\l)t}\hat\xi(t),
\]
where the integral with respect to $t$ is actually $L^1$ and has meaning pointwise for $\theta + i\l$,
$-\pi < \l < 0$ \cite[Theorem III.2.3]{SW71}.
For the convenience of the reader, we give an elementary proof of these facts in Appendix \ref{fourier}
(c.f.\! \cite[Theorem IX.13]{RSII}).
Therefore, the Hardy space $\hardy$ can be considered as a (dense) subspace of
the Hilbert space $L^2(\RR)$.

Next, let us consider the operator of analytic continuation:
\begin{align*}
 \dom(\shift) &:= \hardy, \\
 (\shift\xi)(\theta) &:= \xi(\theta - \pi i).
\end{align*}
This can be identified with the Fourier transform of the multiplication operator
$M_{e_{\pi}}$ by $e_{\pi}(t) = e^{\pi t}$, and it is self-adjoint on this space (see Appendix \ref{fourier}).

\subsection{The bound state operator}
Let $f \in H^\infty(\SS_{-\pi,0})$, the space of bounded analytic function on $\SS_{-\pi,0} = \RR + i(-\pi,0)$.
We do not assume the continuity on the closure $\RR + i[-\pi, 0]$. However,
one can identify the strip $\RR + i(-\pi, 0)$ with the unit disk in $\CC$ by
a conformal transformation (see, e.g.\! \cite[Appendix A]{LW11}) and it follows
from the boundedness of $f$ that $f(\zeta)$ has radial boundary values
when $\im \zeta \to 0,-\pi$, namely, $f(\theta + i\l)$ converges when $\l \to 0,-\pi$
for almost every $\theta$ \cite[Theorem 11.32]{Rudin87}.
We denote these boundary values by $f(\theta)$ and $f(\theta - \pi i)$, respectively, which are $L^\infty(\RR)$.
We further assume the property that $\overline{f(\theta)} = f(\theta - \pi i)$ almost everywhere.
Let $\mfbar$ be the multiplication operator by $\overline{f(\theta)} = f(\theta - \pi i)$.

Our main object in this paper is the operator $\mfbar\shift$.
This product operator is closable by the following Lemma (see \cite[Theorem 13.2]{Rudin91}).
\begin{lemma}\label{lm:star}
 Let $x$ be a bounded operator, $A$ be a closed operator such that
 $A^*x^*$ is densely defined. Then $xA$ is closable and it holds that $(xA)^* = A^*x^*$.
\end{lemma}
It holds that $M_f^* = \mfbar$.
It is easy to see that $\shift \mfbar^* = \shift M_f$ is densely defined.
Indeed, $M_f$ preserves the domain $H^2(\SS_{-\pi,0})$ of $\shift$.
With the symmetry condition $\overline{f(\theta)} = f(\theta - \pi i)$, we can say more \cite[Proposition 3.1]{CT15-1}.
\begin{proposition}\label{pr:symmetric}
The operator $\mfbar\shift$ is symmetric.
\end{proposition}
Let us briefly recall the proof. Take $\xi,\eta \in \dom(\shift)$ which have compact spectral supports and are smooth.
The function $\bar\eta(\theta) := \overline{\eta(\bar\theta)}$ belongs to $H^\infty(\SS_{0,\pi})$ and
\begin{align*}
 \<\eta, \mfbar\shift\xi\> &= \int \overline{\eta(\theta)}f(\theta - \pi i)\xi(\theta - \pi i)d\theta \\
 &= \int \bar\eta(\theta + \pi i) f(\theta) \xi(\theta)d\theta \\
 &= \int \overline{\eta(\theta - \pi i) \overline{f(\theta)}} \xi(\theta)d\theta \\
 &= \<\mfbar\shift\eta, \xi\>,
\end{align*}
where in the second line we used Cauchy's theorem and rapid decay of $\xi$ and $\eta$.
By continuity, this equation holds for any pair $\xi,\eta \in \dom(\shift)$, therefore,
we obtain the symmetry.
In particular, one obtains again that the product operator $\mfbar\shift$ is closable.

This notation of $\Delta$ is (almost) compatible with its use in integrable two-dimensional quantum field theory,
which we will investigate in \cite{Tanimoto15-2}.
Namely, $\Delta$ is the one-particle component of the modular operator for the von Neumann algebra
\cite{TakesakiII} corresponding to the right standard wedge with respect to the vacuum vector
(which is sometimes denoted by $\Delta_1$, but we omit the subscript for simplicity).
The unitary operators $\Delta^{it}$ coincide with the one-particle action of the Lorentz boosts
in many cases (see \cite{BL04} for the integrable models without bound states) and gives
the shift $\xi(\theta+2\pi t)$, therefore,
$\shift$ is the analytic continuation $\theta \mapsto \theta - \pi i$.

The product operator $\mfbar\shift$ appears in the study of quantum
field theories with bound states\cite{CT15-1}. One of the principal problems in quantum field theory is to
construct local observables (a class of self-adjoint operators). For a family of integrable
quantum field theories in two spacetime dimensions,
we constructed a candidate of observables localized in a wedge-shaped region.
This candidate operator contains $\mfbar\shift$ as a building block.
More precisely, it is the one-particle component of the operator which makes
a bound state.
Yet, its correct self-adjoint domain and locality in a strong sense remained open.
As these properties are crucial in constructing Haag-Kastler nets (operator-algebraic
realization of quantum field theories), we investigate the self-adjointness
of $\mfbar\shift$ in this paper.

\subsubsection*{Self-adjointness criterion}
Now the question is whether $\mfbar\shift$ has a self-adjoint extension.
For this purpose, let us recall the fundamental criterion for self-adjointness \cite[Section X.1]{RSII}.

For a symmetric operator $A$, its adjoint $A^*$ may have nonzero eigenvectors for eigenvalues $\pm i$.
We denote the dimensions of the corresponding eigenspace by $n_\pm(A)$.
The pair $(n_+(A), n_-(A))$ is called the {\bf deficiency indices} of $A$.
Equivalently, they are dimensions of the spaces $\ker (A^* \mp i)$.
$A$ can be extended to a self-adjoint extension if and only if $n_+(A) = n_-(A)$
and there is a one-to-one correspondence between such self-adjoint extensions and
isometric operators from $\ker (A^* - i)$ to $\ker (A^* + i)$.
If the deficiency indices of $A$ is $(0,0)$, then $A$ is essentially self-adjoint
and the closure $\overline A$ of $A$ is the unique self-adjoint extension of $A$.

We have $(\mfbar\shift)^* = \shift M_f$.
Therefore, our task is reduced to studying the eigenspaces of $\shift M_f$.

\section{Computing deficiency indices}\label{di}
\subsection{Common ideas}\label{ideas}
Any element $f \in H^\infty(\SS_{-\pi,0})$ admits the following factorization
\cite[Beurling factorization, Theorems 17.15, 17.17]{Rudin87}
(see also the identification between $H^\infty(\SS_{-\pi,0})$ and the Hardy space on the unit circle
\cite[Appendix A]{LW11}):
\begin{align*}
 f(\zeta) &= cf_\blaschke(\zeta) f_\tin(\zeta) f_\tout(\zeta), \\
 f_\blaschke(\zeta) &= \prod_n c_n\frac{e^\zeta - e^{\a_j}}{e^\zeta - e^{\overline{\a_j}}}, \\
 f_\tin(\zeta) &= \exp\left(i\int \frac{d\mu(s)}{1+s^2}\, \frac{1+e^\zeta s}{e^\zeta - s}\right), \\
 f_\tout(\zeta) &= \exp\left(\frac i\pi\int \frac{ds}{1+s^2} \frac{1+e^\zeta s}{e^\zeta - s}\log\phi(s)\right),
\end{align*}
where $\a_j \in \strip$ and satisfies the Blaschke condition which assures the convergence
of the infinite product (see \cite[Theorem 15.21]{Rudin87} for the condition written in
the unit circle picture), $c_n = -\frac{|\b_n|}{\b_n}\frac{\b_n-1}{\overline{\b_n}-1}$,
where $\beta_n = \frac{e^{\a_n}+i}{e^{\a_n}-i}$ (if $\a_j = -\frac{\pi i}2$, we set $c_j = 1$
as a convention).
$c$ is a constant with $|c| = 1$.
$\frac{\mu(s)}{1+s^2}$ is a finite singular measure (with respect to the Lebesgue measure $ds$)
on $\RR \cup \{\infty\}$ and may have an atom at $\infty$ ($-\infty$ is identified with $\infty$).
$\phi(s)$ is a positive function on $\RR$ such that $\frac{\log\phi(s)}{1+s^2}$ is $L^1(\RR)$, again with respect to
the Lebesgue measure $ds$.
$f_\blaschke$ and $f_\tin$ are {\bf inner}, namely
$|f_\blaschke(\theta)| = |f_\blaschke(\theta - \pi i)| = |f_\tin(\theta)| = |f_\tin(\theta - \pi i)| = 1$
for almost every $\theta \in \RR$ (they are defined as the boundary values and not as integrals
for $\zeta = \theta \in \RR, \theta - \pi i$, which might be meaningless).
On the other hand, $f_\tout$ is said to be {\bf outer}
and it is the exponential of the Poisson integral of the kernel $\frac{\log\phi(s)}{1+s^2}$.

This decomposition is unique. We call $f_\blaschke$ the Blaschke product of $f$,
$f_\tin$ the singular inner part of $f$, and $f_\tout$ the outer part of $f$.

In order to solve the eigenvalue equation $\shift M_f\xi = \pm i\xi$,
it is helpful to extend the domain of consideration to meromorphic functions
on the strip $\RR + i(-\pi, 0)$. Once we find a meromorphic function $g$
which satisfies $f(\theta - \pi i)g(\theta - \pi i) = \pm ig(\theta)$,
any solution $\xi$ in the domain of the operator $\shift M_f$ can be divided
by $g$ and one obtains a periodic function:
\[
 \frac{\xi(\theta - \pi i)}{g(\theta - \pi i)} = \frac{f(\theta - \pi i)\xi(\theta - \pi i)}{f(\theta - \pi i)g(\theta - \pi i)}
 = \frac{\pm i \xi(\theta)}{\pm i g(\theta)} = \frac{\xi(\theta)}{g(\theta)}.
\]
Some properties of this periodic function can be derived from those of
$\xi, f$ and $g$, and we might be able to classify such periodic functions.
We will call the equation $f(\theta - \pi i)g(\theta - \pi i) = \pm ig(\theta)$
for simplicity the \textbf{eigenvalue equation} with eigenvalue $\pm i$.
Furthermore, the question is equivalent to finding a solution for the eigenvalue $1$,
because by multiplying $e^{\a \theta}$ one can vary the eigenvalue by $e^{-i\a \pi}$.
We will see later how this idea applies to concrete cases.

Let us prepare a Lemma which is useful in such an argument.
The following is a slight variation of \cite[Section 2, Theorem]{Rudin71}.
\begin{lemma}\label{lm:local}
 Let $h$ be a measurable function defined on a rectangle $[\theta_1,\theta_2] + i[\l_1,\l_2] \subset \CC$
 where $\l_1 < 0 < \l_2$, horizontally $L^2$, namely $h(\,\cdot + i\l)$ is in $L^2([\theta_1, \theta_2])$
 for each $\lambda$
 and assume that its $L^2$-norm is uniformly bounded in $\l$, and that $\l \mapsto h(\,\cdot + i\l)$
 is continuous from $\RR$ to $L^2([\theta_1,\theta_2])$.
 Furthermore, assume that $h$ is analytic separately in $(\theta_1,\theta_2) + i(\l_1,0)$ and in
 $(\theta_1,\theta_2) + i(0,\l_2)$. Then $h$ has a representative (in the sense of pointwise function, as
 $L^2$-functions are not pointwise defined) which is analytic on the whole rectangle.
\end{lemma}
\begin{proof}
 As the $L^2$-norm of $h(\,\cdot + i\l)$ is uniformly bounded in $\l$, by Fubini's theorem
 \cite[Theorem 8.8]{Rudin87}, $h$ is $L^2$ as a two-variable function, and again by Fubini's theorem,
 $h(\theta + i\,\cdot)$ is in $L^2([\l_1,\l_2])$ for almost all $\theta \in [\theta_1,\theta_2]$.
 
 We only have to prove the analyticity at $\l = 0$. For a given $\theta$, $\theta_1 < \theta < \theta_2$,
 we take $\theta_3 < \theta < \theta_4$ such that
 $h(\theta_3 + i\,\cdot)$ and $h(\theta_4 + i\,\cdot)$ are in $L^2([\l_1,\l_2])$. Consider the rectangle with the corners
 $\theta_3 + i\l_1,\theta_4 + i\l_1,\theta_4 + i\l_2,\theta_3 + i\l_2$ and take a counterclockwise path $\G$.
 Then
 \[
  \widetilde h(\zeta) := \int_\G dz\, \frac{h(z)}{z-\zeta}
 \]
 defines an analytic function inside the rectangle,
 as the integrand is ($L^2$ on the bounded segments, therefore) $L^1$.
 
 Let us see that $\widetilde h$ coincides with $h$ on the upper- and lower-half rectangles.
 Indeed, let $0 < \epsilon < \im \zeta$ and take a rectangular path $\G_\epsilon$ with the corners
 $\theta_3 + i\epsilon, \theta_4 + i\epsilon, \theta_4 + i\l_2, \theta_3 + i\l_2$ which contains $\zeta$ in its
 inside (see Figure \ref{fig:periodicity-contour}). By Cauchy's formula, we have
 \[
  h(\zeta) = \int_{\G_\epsilon} dz\, \frac{h(z)}{z-\zeta}.
 \]
 As $\epsilon \to 0$, this integral tends to
 \[
  \int_{\G_+} dz\, \frac{h(z)}{z-\zeta},
 \]
 where $\G_+$ is the rectangle with corners $\theta_3, \theta_4, \theta_4 + i\l_2, \theta_3 + i\l_2$,
 by the assumption of continuity of $h$ (in the $L^2$, therefore) in the $L^1$ sense
 and the $L^1$-integrability of $h$ on the sides of the rectangle.
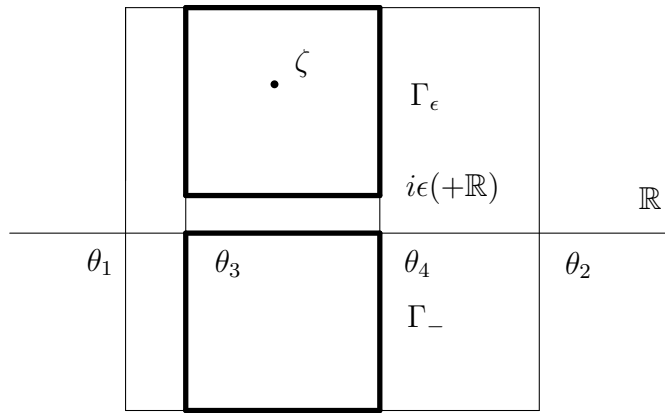
\begin{figure}[ht]
    \centering
\begin{tikzpicture}[line cap=round,line join=round,>=triangle 45,x=1.0cm,y=1.0cm]
\clip(-4.3,0.5) rectangle (4.4,6.3);
\draw [domain=-4.3:7.4] plot(\x,{(--21.77-0*\x)/7.56});
\draw (-2.76,5.88)-- (-2.76,0.52);
\draw (-2.76,0.52)-- (2.74,0.52);
\draw (2.74,0.52)-- (2.74,5.88);
\draw (-2.76,5.88)-- (2.74,5.88);
\draw (-2.76,5.88)-- (-2.76,3.58);
\draw (3.92,3.64) node[anchor=north west] {$\mathbb{R}$};
\draw (-0.66,5.48) node[anchor=north west] {$\zeta$};
\draw (-1.72,2.8) node[anchor=north west] {$\theta_3$};
\draw (0.8,2.8) node[anchor=north west] {$\theta_4$};
\draw (-3.42,2.82) node[anchor=north west] {$ \theta_1 $};
\draw (2.94,2.78) node[anchor=north west] {$\theta_2$};
\draw (-1.96,5.88)-- (-1.96,0.52);
\draw (-1.96,0.52)-- (0.62,0.52);
\draw (0.62,0.52)-- (0.62,5.88);
\draw [line width=2pt] (0.62,5.88)-- (-1.96,5.88);
\draw [line width=2pt] (-1.96,5.88)-- (-1.96,3.38);
\draw [line width=2pt] (-1.96,3.38)-- (0.62,3.38);
\draw [line width=2pt] (0.62,3.38)-- (0.62,5.88);
\draw (0.88,5.02) node[anchor=north west] {$\Gamma_\epsilon$};
\draw [line width=2pt] (-1.96,2.88)-- (-1.96,0.52);
\draw [line width=2pt] (-1.96,0.52)-- (0.62,0.52);
\draw [line width=2pt] (0.62,0.52)-- (0.62,2.88);
\draw [line width=2pt] (0.62,2.88)-- (-1.96,2.88);
\draw (0.84,2.06) node[anchor=north west] {$\Gamma_-$};
\draw (0.82,3.84) node[anchor=north west] {$i\epsilon (+\mathbb{R})$};
\begin{scriptsize}
\fill [color=black] (-0.78,4.86) circle (1.5pt);
\end{scriptsize}
\end{tikzpicture}
    \caption{The integral contours for the proof of analyticity}
    \label{fig:periodicity-contour}
\end{figure}

 On the other hand, if we consider a rectangle in a lower half-plane, the integral
 gives $0$ by Cauchy's theorem. By a similar continuity argument, we obtain
 \[
  \int_{\G_-} dz\, \frac{h(z)}{z-\zeta} = 0,
 \]
 where $\G_-$ is the rectangle with corners $\theta_3, \theta_3 + i\l_1, \theta_4 + i\l_1, \theta_4$.
 Therefore, altogether we get
 \[
  h(\zeta) =  \int_{\G_+ \cup \G_-} dz\, \frac{h(z)}{z-\zeta} = \int_\G dz\, \frac{h(z)}{z-\zeta} = \widetilde h(\zeta).
 \]
 By a parallel argument, we have $h(\zeta) = \widetilde h(\zeta)$ for $\im \zeta < 0$. In other words,
 $h$ has an analytic extension to the whole rectangle.
 By the $L^2$-continuity, it must hold that $h=\widetilde h$.
\end{proof}

We use this Lemma in the following form.
\begin{proposition}\label{pr:period}
 Suppose that $h$ is analytic in $\RR + i(-\pi, 0)$ and
 it holds that $h(\theta) = h(\theta - \pi i)$ in the local $L^2$ sense,
 namely, for a finite interval $[\theta_1,\theta_2]$
 $h(\cdot - \epsilon i) \to h(\cdot)$ and 
 $h(\cdot - (\pi - \epsilon)i) \to h(\cdot - \pi i)$ in $L^2([\theta_1,\theta_2])$
 and the limits coincide.
 Then $h$ extends to an analytic function with period $\pi i$.
\end{proposition}
\begin{proof}
 Define the periodic function by $h(\zeta + N\pi i) = h(\zeta), N \in \NN$.
 The only question is the analyticity at $\zeta \in \RR + N\pi i$, which is
 a direct consequence of Lemma \ref{lm:local}.
\end{proof}

\subsection{Finite Blaschke products}\label{finiteblaschke}
Let
\[
 f(\zeta) = \prod_{j=1}^n c_j\frac{e^\zeta - e^{\a_j}}{e^\zeta - e^{\overline{\a_j}}}
\]
be a finite Blaschke product, where $\a_j \in \strip$. Note that they have exactly $n$-zeros (including
multiplicity) in the strip $\RR + i(-\pi,0)$, since the function $e^\zeta$ is one-to-one in the strip.
In order that this satisfies $\overline{f(\theta)} = f(\theta - \pi i)$,
$\a_j$ and $\overline{\a_j} - \pi i$ must appear in pair including multiplicity,
or $\im \a_j = -\frac \pi 2$. Under this condition, $c_j$'s cancel each other or it is $1$ or $-1$.
$-1$ does not affect neither the domain property nor symmetry, therefore, we may omit $c_j$
and write $f(\zeta) = \prod_{j=1}^n \frac{e^\zeta - e^{\a_j}}{e^\zeta - e^{\overline{\a_j}}}$

We consider the operator $\mfbar\shift$ and compute the deficiency indices.
Let us start with a Lemma concerning periodic functions.
\begin{lemma}\label{lm:polynom}
 If $h$ is analytic in $\strip$, $\pi i$-periodic, namely $h(\theta - \pi i) = h(\theta)$,
 and satisfies $|h(\zeta)| \le A e^{N |\re\zeta|}$, where $A > 0$ and $N$ is a positive even integer,
 then there is a polynomial $p$ of degree
 less than or equal to $N$ such that $h(\zeta) = p(e^{2\zeta})e^{-N\zeta}$.
\end{lemma}
\begin{proof}
 Consider $h_\vdash(\zeta) = h(\zeta)e^{N\zeta}$, which satisfies
 $|h_\vdash(\zeta)| \le A(e^{2N\re \zeta} + 1)$. By assumption, $N$ is even, therefore,
 $h_\vdash(\zeta)$ is again periodic by $\pi i$.
 
 We can define an analytic function on $\CC \setminus \{0\}$ by
 $h_\vdash\left(\frac{\log z}2\right)$, which is well-defined by the $\pi i$-periodicity of $h$.
 The bound of $h_\vdash$ can be translated into
 \[
  \left|h_\vdash\left(\frac{\log z}2\right)\right| \le A (|z|^N + 1),
 \]
 and especially, the singularity at $z = 0$ is removable. Then it is a well-known consequence
 of the Cauchy estimate \cite[Theorem 10.26]{Rudin87} that $h_\vdash\left(\frac{\log z}2\right)$
 is a polynomial $p(z)$ of degree less than or equal to $N$.
 
 Since $\zeta = \frac{\log z}2$ or $z = e^{2\zeta}$, we have $h(\zeta) = p(e^{2\zeta})e^{-N\zeta}$ as desired.
\end{proof}

\begin{proposition}\label{pr:bl-fin}
 If $f$ contains $2m$ factors, respectively $2m + 1$ factors, where $m$ is an integer,
 then the deficiency indices of $\mfbar\shift$ is $(m,m)$, respectively $(m+1,m)$.
\end{proposition}
\begin{proof}
 We observed after Lemma \ref{lm:star} that $(\mfbar\shift)^* = \shift M_f$.
 By definition of the product of (possibly) unbounded operators, the domain of $\shift M_f$ is:
 \[
  \{\xi \in L^2(\RR^2)\mid f(\zeta)\xi(\zeta) \in \hardy\}.
 \]
 
 In order to determine the deficiency indices, we have to find the eigenvectors of $\shift M_f$ corresponding
 to the eigenvalues $\pm i$. We claim that the functions
 \[
  \xi_k(\theta) = e^{\left(k+\frac12\right)\theta}\prod^n_{j=1} \frac{1}{e^\theta - e^{\a_j}},
 \]
 where $0 \le k < n$, are precisely those eigenfunctions. It is easy to see that they have $n$-poles in the
 strip $\RR + i(-\pi,0)$.
 
 Firstly, let us show that they are indeed eigenfunctions. We have
 \[
  f(\theta)\xi_k(\theta) = e^{\left(k+\frac12\right)\theta}\prod^n_{j=1} \frac{1}{e^\theta - e^{\overline{\a_j}}}
 \]
 We observe that this has no pole in $\strip$, since $e^\zeta$ takes only values with negative imaginary part,
 while $e^{\overline{\a_j}}$ have positive imaginary parts. The product $f(\theta)\xi_k(\theta)$
 belongs to $\hardy$:
 \[
  |f(\zeta)\xi_k(\zeta)| \le e^{\left(k+\frac12\right)\re\zeta}\prod^n_{j=1} \left|\frac{1}{e^{\re \zeta} - e^{\overline{\a_j}}}\right|,
 \]
 again because $e^\zeta$ has negative imaginary part, while $e^{\overline{\a_j}}$ have positive imaginary parts.
 The right-hand side decays exponentially when $\re\zeta \to \pm \infty$, since $0 \le k < n$,
 therefore, $f(\zeta)\xi_k(\zeta) \in \hardy$.
 The boundary value can be straightforwardly computed and it is
 \begin{align*}
  f(\theta-\pi i)\xi_k(\theta-\pi i) &= (-1)^{k+1} ie^{\left(k+\frac12\right)\theta}\prod^n_{j=1} \frac{(-1)^n}{e^\theta + e^{\overline{\a_j}}} \\
    &= (-1)^{k+1} ie^{\left(k+\frac12\right)\theta}\prod^n_{j=1} \frac{(-1)^n}{e^\theta - e^{\a_j}} \\
    &= (-1)^{n+k+1}i\xi_k(\theta)
 \end{align*}
 where we used that $\a_j$ and $\overline{\a_j}-\pi i$ appear in pair, or $\im \a_j = \frac \pi 2$
 and $e^{\a_j} = -e^{\overline{\a_j}}$.
 This means that this belongs to the deficiency subspace. More precisely, if $k = 0$ and if $n = 2m$,
 it belongs to $n_-(\mfbar\shift)$ and if $n = 2m+1$, to $n_+(\mfbar\shift)$.
 As $k$ increases, it alters the eigenvalue by $-1$. They are obviously linearly independent.

 Summing up, if $m$ is even, we found $m$ vectors both in $n_\pm(\mfbar\shift)$ and if $m$ is odd,
 we found $m+1$ vectors in $n_+(\mfbar\shift)$ and $m$ vectors in $n_+(\mfbar\shift)$.
 
 Secondly, we claim that their linear spans exhaust the deficiency subspaces.
 Let $\xi$ be in $\ker (\shift M_f \mp i)$. Recall that $\xi_k$ are also eigenfunctions
 of $\shift M_f$, to which eigenspace it belongs depends on $k$.
 As described in Section \ref{ideas}, we need just one function which satisfies the
 eigenvalue equation, which is not necessarily in the Hilbert space.
 The restriction $0 \le k < n$ is required only when one wants a Hilbert space vector,
 and for other $k$, $\xi_k$, defined analogously, satisfies the same eigenvalue equation.
 We take a $\xi_k$ which has the same eigenvalue as $\xi$.
 
 We proceed as we described in Section \ref{ideas}:
 We consider $\frac{\xi(\zeta)}{\xi_k(\zeta)} = \frac{f(\zeta)\xi(\zeta)}{f(\zeta)\xi_k(\zeta)}$.
 By assumption, $f\xi \in \hardy$. On the other hand, we have
 \[
  \left|\frac1{f(\zeta)\xi_k(\zeta)}\right| = \left|e^{-\left(k + \frac12\right)\zeta}\prod_{j=1}^n (e^\zeta - e^{\overline{\a_j}})\right|
  \le Ae^{M|\re \zeta|},
 \]
 where $A$ and $M$ are positive constants. Especially, this factor has no pole.
 Therefore, the function $\frac{\xi(\zeta)}{\xi_k(\zeta)}$ is analytic in $\RR +i(-\pi,0)$
 and has the same locally $L^2$-boundary value at $\im \zeta = 0, -\pi$,
 hence by Proposition \ref{pr:period}, it extends to an entire periodic function.
 An entire periodic function admits a Fourier expansion: $\frac{\xi(\zeta)}{\xi_k(\zeta)} = \sum a_j e^{j\zeta}$,
 where $a_j$ is strongly decreasing, and this sum is uniformly convergent on any compact set
 with respect to $\re \zeta$.
 
 For any $L^2$-function $\eta$ supported in $[-R,R]$, the following integral
 \[
  \int d\theta\, \overline{\eta(\theta)} \frac{\xi(\zeta+\theta)}{\xi_k(\zeta+\theta)}
 = \int d\theta\, \overline{\eta(\theta)} \frac{f(\zeta+\theta)\xi(\zeta+\theta)}{f(\zeta+\theta)\xi_k(\zeta+\theta)}
 \]
 defines an analytic function of $\zeta$ in $\strip$ by Morera's theorem \cite[Theorem 10.17]{Rudin87}.
 Furthermore, we have the following estimate:
 \begin{align*}
   &\left|\int d\theta\, \overline{\eta(\theta)} \frac{f(\zeta+\theta)\xi(\zeta+\theta)}{f(\zeta+\theta)\xi_k(\zeta+\theta)}\right| \\
 \le\; & Ae^{M(|\re \zeta|+R)}\left|\int_{-R}^R d\theta\, \overline{\eta(\theta)} f(\zeta+\theta)\xi(\zeta+\theta)\right| \\
 \le\; & Ae^{M(|\re \zeta|+R)}\|\eta\|\cdot \|f\xi\|_{\hardy}.
 \end{align*}
 This function is periodic since so is $\frac{\xi(\zeta)}{\xi_k(\zeta)}$,
 hence by Lemma \ref{lm:polynom}, it must be of the form $\sum_{j = -N}^N a_{\eta,j} e^{j\zeta}$,
 where $N$ is the smallest even integer such that $N > M$.
 
 By the compactness of the support of $\eta$ and the uniform convergence of
 $\frac{\xi(\zeta)}{\xi_k(\zeta)} = \sum a_j e^{j\zeta}$, we obtain
 $a_{\eta, j} = a_j\int d\theta\, \overline{\eta(\theta)} e^{j\theta}$.
 As $\eta$ is arbitrary, this implies that $a_j = 0$ for $|j| > N$.
 In other words, $\xi(\zeta) = \xi_k(\zeta) \sum_{-N}^N a_j e^{j\zeta}$.
 
 This is actually of the form $\xi(\zeta) = \sum_{-N}^N a_j \xi_{k+j}(\zeta)$,
 but we know that $\xi_{k+j}(\zeta) $ can be in $\hardy$ if and only if
 $0 \le k+j < n$, and their decay rates are different for different $k+j$,
 therefore, $\xi(\zeta) \in \hardy$ if and only if
 $\xi(\zeta) = \sum_{j=0}^{n-1} a_j \xi_{j}(\zeta)$. This completes the proof that $\{\xi_j\}_{j=0}^{n-1}$
 exhaust the deficiency subspaces.
\end{proof}

\subsection{Infinite Blaschke products}\label{infiniteblaschke}
From the results of the previous section, it is natural to expect that
for an infinite Blaschke product
\[
 f(\zeta) = \prod_{j=1}^\infty c_j\frac{e^\zeta - e^{\a_j}}{e^\zeta - e^{\overline{\a_j}}}
\]
the deficiency indices $n_\pm(\mfbar\shift)$ are infinite.
This will turn out to be true, yet it is not easy in general to calculate
the deficiency subspaces explicitly, as the natural candidate vectors in those
subspaces would be infinite products, whose convergence is not always under control.
Therefore, we divide the cases.

\subsubsection{Zeros on the middle line}
Let us assume that $\im \a_j = -\frac\pi 2$. In this case, $e^{\a_j} =: i\g_j \in i\RR_-$.
We divide the negative numbers $\{\g_j\}$ into two families according to their real parts and
reorder as follows: $\re \a^+_j \ge 0$ and $\re \a^-_j < 0$, accordingly
$\g^+_j \ge -1$ and $0 < \g^-_j < -1$.
Correspondingly, we rename the constant factors, and then by definition: $c^+_j = -1,
c^-_j = 1$.

\begin{proposition}
 Let $f$ be a Blaschke product with infinitely many zeros $\{\a_j\}$ on the line $\RR + \frac{\pi i}2$.
 Then $n_\pm(\mfbar\shift) = \infty$.
\end{proposition}
\begin{proof}
As the infinite product is absolutely convergent, from \cite[Theorem 15.5]{Rudin87}, it follows
for $\zeta \in \strip$ that
\begin{align*}
 \sum_j\left|1-c^+_j\frac{e^\zeta - i\g^+_j}{e^\zeta + i\g^+_j}\right| &= \sum_j\left|\frac{2e^\zeta}{e^\zeta + i\g^+_j}\right| < \infty, \\
 \sum_j\left|1-c^-_j\frac{e^\zeta - i\g^-_j}{e^\zeta + i\g^-_j}\right| &= \sum_j\left|\frac{i2\g^-_j}{e^\zeta + i\g^-_j}\right| < \infty.
\end{align*}

By noting that $\g^+_j \to -\infty, \g^-_j \to 0$, for a fixed $\zeta = a+ib\in \strip$, there is sufficiently large $j$ such that
$2|b| \le |\g^+_j|$ and consequently $\left|\frac{a}{a + i\g^+_j}\right| \le \left|\frac{2e^\zeta}{e^\zeta + i\g^+_j}\right|$,
and $2|\g^-_j| \le |b|$ and consequently
$\left|\frac{i\g^-_j}{a + i\g^-_j}\right| \le \left|\frac{i2\g^-_j}{e^\zeta + i\g^-_j}\right|$.
By comparing with the above convergent series, 
it is straightforward to see that the following infinite sums are convergent also for real $\theta$ such that $e^\theta = a$:
\begin{align*}
 \sum_j\left|1-\frac{-i\g^+_j}{e^\theta - i\g^+_j}\right| &= \sum_j\left|\frac{e^\theta}{e^\theta + i\g^+_j}\right|, \\
 \sum_j\left|1-\frac{e^\theta}{e^\theta - i\g^-_j}\right| &= \sum_j\left|\frac{i\g^-_j}{e^\theta + i\g^-_j}\right|. 
\end{align*}
Therefore, the following infinite products
\[
 \prod_j\frac{-i\g^+_j}{e^\theta - i\g^+_j}, \;\;\; \prod_j\frac{e^\theta}{e^\theta - i\g^-_j}
\]
are absolutely convergent and non zero for $\theta \in \RR$.

We consider the following function
\[
 \xi_k(\theta) := e^{(k+\frac12)\theta}\prod_j\frac{-i\g^+_j}{e^\theta - i\g^+_j} \prod_j\frac{e^\theta}{e^\theta - i\g^-_j}.
\]
We claim that $\xi_k \in L^2(\RR,d\theta)$ for infinitely many $k \in \ZZ$.
Indeed, each factor in the big products has the modulus smaller than $1$,
the former factors decay exponentially as $\theta \to \infty$, while the latter factors
decay exponentially as $\theta \to -\infty$.
By assumption, at least one of these products is infinite, hence it can decay faster than $e^{(k+\frac12)\theta}$
for $k > 0$ or $k < 0$, depending on which is infinite.

Next, we claim that $f(\zeta)\xi_k(\zeta) \in \hardy$. By $c_+ = -1$ and $c_- = 1$, it is easy to check that
\[
 f(\zeta)\xi_k(\zeta) = e^{(k+\frac12)\zeta}\prod_j\frac{i\g^+_j}{e^\zeta + i\g^+_j} \prod_j\frac{e^\zeta}{e^\zeta + i\g^-_j}
\]
and this is dominated by $\left|e^{(k+\frac12)\re \zeta}f(\re \zeta)\xi(\re \zeta)\right|$.
Finally,
\begin{align*}
 f(\theta - \pi i)\xi_k(\theta - \pi i)
 &= (-1)^{k+1}ie^{(k+\frac12)\theta}\prod_j\frac{i\g^+_j}{-e^\theta + i\g^+_j} \prod_j\frac{-e^\theta}{-e^\theta + i\g^-_j} \\
 &=(-1)^{k+1}ie^{(k+\frac12)\theta}\prod_j\frac{-i\g^+_j}{e^\theta - i\g^+_j} \prod_j\frac{e^\theta}{e^\theta - i\g^-_j} \\
 &= (-1)^{k+1}i\xi_k(\theta).
\end{align*}
Therefore,
$\xi_k(\theta) \in \ker(\shift M_f \mp i)$,
where $-$ applies when $k$ is odd and $-$ applies when $k$ is even.
In particular, there are infinitely many such eigenvectors.
\end{proof}

\subsubsection{Zeros outside the middle line}\label{zerosoutside}
Here we assume that $\im \a_j \neq -\frac \pi 2$.
By the symmetry condition, $\a_j$ must appear with $\overline{\a_j}-\pi i$ in pair.
Let us reorder and rename the zeros $\{\a_j\}$ and write this explicitly:
\[
 f(\zeta) = \prod_{j=1}^\infty \frac{e^\zeta - e^{\a_j}}{e^\zeta - e^{\overline{\a_j}}}\frac{e^\zeta + e^{\overline{\a_j}}}{e^\zeta + e^{\a_j}}.
\]
Note that $f$ is inner, hence the multiplication operator $\mfbar$ is unitary and
the product $\mfbar\shift$ is a closed symmetric operator.

We first show that $\mfbar\shift$ has a self-adjoint extension. Indeed,
using the above reordering,
we have $f(\zeta) = f_+(\zeta)f_-(\zeta)$, where
\[
 f_+(\zeta) = \prod_{j=1}^\infty \frac{e^\zeta + e^{\overline{\a_j}}}{e^\zeta + e^{\a_j}}, \;\;\;
 f_-(\zeta) = \prod_{j=1}^\infty \frac{e^\zeta - e^{\a_j}}{e^\zeta - e^{\overline{\a_j}}}.
\]
Both functions are in $H^\infty(\SS_{-\infty,0})$ and it holds that
$\overline{f_+(\theta - \pi i)} = f_-(\theta)$.

By definition of the domains of product operators,
we have $M_{f_-(\,\cdot - \pi i)}\shift \subset \shift M_{f_-}$, since
$f_-(\zeta)$ is bounded and analytic in $\strip$.
Therefore, we have the following inclusion
\[
 \mfbar\shift = M_{f_+(\,\cdot - \pi i)}M_{f_-(\,\cdot - \pi i)}\shift
 \subset M_{f_+(\,\cdot - \pi i)}\shift M_{f_-} = M_{f_-}^*\shift M_{f_-}.
\]
The last expression is a self-adjoint operator, since $M_{f_-}$ is unitary.
In particular, the operator $\mfbar\shift$ has a self-adjoint extension
and $n_+(\mfbar\shift) = n_-(\mfbar\shift)$.

Between $\mfbar\shift$ and $M_{f_-}^*\shift M_{f_-}$, there are infinitely many
different symmetric closed operators. Indeed,
let us put
$f_{j,+}(\zeta) := \frac{e^\zeta + e^{\overline{\a_j}}}{e^\zeta + e^{\a_j}}$,
$f_{j,-}(\zeta) := \frac{e^\zeta - e^{\a_j}}{e^\zeta - e^{\overline{\a_j}}}$.
We can see as above
\begin{align*}
 &\mfbar\shift \subset \mfbar M_{f_{1,-}}^*\shift M_{f_{1,-}}
 \subset \cdots \subset \mfbar M_{f_{j,-}}^*\cdots M_{f_{1,-}}^*\shift M_{f_{1,-}} \cdots M_{f_{j,-}} \\
 &\cdots \subset M_{f_-}^*\shift M_{f_-}.
\end{align*}
Such an infinite tower of extensions is possible only if $n_+(\mfbar\shift) = n_-(\mfbar\shift) = \infty$.

\subsection{Singular inner functions}\label{inner}
In this Subsection we consider singular inner functions.
An singular inner function admits the following representation
\[
 f(\zeta) = \exp\left(i\int \frac{d\mu(s)}{1+s^2}\, \frac{1+e^\zeta s}{e^\zeta - s}\right),
\]
where $\mu$ is a singular measure on $\RR \cup \{\infty\}$.
If $\mu$ has atoms at $0, \infty$, we need a different treatment. Let us consider these cases
separately.

\subsubsection{Atomic measures at infinity}\label{atomic}
When $\frac{\mu(\{0\})}{1+s^2} = \a \ge 0$ and $\frac{\mu(\{\infty\})}{1+s^2} = \beta \ge 0$ and $\mu = 0$ elsewhere,
our function takes the form $f(\zeta) = \exp\left(i\a e^{-\zeta} - i\beta e^{\zeta}\right)$.
As in Section \ref{ideas}, we look for solutions of the eigenvalue equation
$f(\theta - \pi i)\xi(\theta - \pi i) = \pm i \xi(\theta)$.
One such solution is $g(\zeta) = \exp\left(-\frac{i\a}{2} e^{-\zeta} + \frac{i\beta}2 e^{\zeta}\right)$.
Indeed, it is straightforward that
\[
 f(\zeta - \pi i) = \exp\left(-i\a e^{-\zeta} + i\beta e^{\zeta}\right),\;\; g(\zeta - \pi i) = \exp\left(\frac{i\a}{2} e^{-\zeta} - \frac{i\beta}2 e^{\zeta}\right) 
\]
and we have $f(\theta - \pi i)g(\theta - \pi i) = g(\theta), \theta \in \RR$.
Furthermore, note that
\[
 |f(\zeta)g(\zeta)| = \left|\exp\left(-\frac{i\a}2 e^{-\zeta} + \frac{i\beta}2 e^{\zeta}\right)\right|
\]
and the problem is reduced to find $\pi i$-periodic functions
with a certain growth condition.

\begin{lemma}\label{lm:pl-l2}
 Let $h$ be continuous on $\RR + [-\pi,0],$ analytic in $\RR + i(-\pi, 0)$ and suppose that
 there is $A, B > 0, 0 < \a < 1$ such that
 \[
  |h(\zeta)| \le Ae^{Be^{\a|\re \zeta|}},
 \]
 and on the boundary, $h(\theta), h(\theta - \pi i)$ are $L^2$.
 Then it follows that $h \in \hardy$.
\end{lemma}
\begin{proof} 
 We argue as in Proposition \ref{pr:bl-fin}.
 For an $L^2$-function $\eta$ supported in $[-R,R]$, the following integral
 \[
  h_\eta(\zeta) := \int d\theta\, \overline{\eta(\theta)} h(\zeta + \theta)
 \]
 defines an analytic function $h_\eta$ of $\zeta$ with the estimate:
 \begin{align*}
   \left|\int d\theta\, \overline{\eta(\theta)} h(\zeta + \theta)\right|
 \le\; & Ae^{B(|\re \zeta|+R)}\left|\int_{-R}^R d\theta\, \overline{\eta(\theta)}\right| \\
 \le\; & Ae^{B(|\re \zeta|+R)}\|\eta\|\cdot \sqrt{2R}
 \end{align*}
and on the boundary, it is immediate from the assumption that
$\max\{|h_\eta(\theta)|, |h_\eta(\theta - \pi i)|\} \le \|\eta\|\cdot \max\{\|h\|, \|h(\,\cdot-\pi i)\|\}$.
Therefore, by Phragm\'en-Lindel\"of principle \cite[Theorem 12.9]{Rudin87},
$|h_\eta(\zeta)| \le \|\eta\|\cdot \max\{\|h\|, \|h(\,\cdot-\pi i)\|\}$.

As this estimate does not depend on $R$, it holds for any $\eta \in L^2(\RR)$,
which implies that $\|h(\,\cdot - i\lambda)\| \le \max\{\|h\|, \|h(\,\cdot-\pi i)\|\}$.
Namely, $h \in \hardy$.
\end{proof}

Next, consider the function $\cos\left(t e^{-\zeta}\right)$, $t \in \RR$.
This is $\pi i$-periodic in $\zeta$ and it holds that
\[
 \left|\cos\left(t e^{-\zeta}\right)\right| \le e^{t |\im e^{-\zeta}|}.
\]
Let $\k$ be a real smooth function supported in $(-\frac\a 3, \frac\a 3)$ such that
$\k(-t) = \k(t)$, $\int dt\,\k(t) = 0$. Then its Fourier transform is entire and we have
\[
\hat \k(e^{-\zeta}) = \frac1{\sqrt{2\pi}}\int dt\, e^{-ite^{-\zeta}}\k(t)
= \frac1{\sqrt{2\pi}}\int_0^\frac{\a}3 dt\, 2\cos(te^{-\zeta})\k(t).
\]
Then we have the following bound:
\[
|\hat \k(e^{-\zeta})| = \frac{\max\{|\k(t)|\}}{\sqrt{2\pi}}\frac{2\a}3 \max\{|\cos(te^{-\zeta})|\}
\le \frac{\max\{|\k(t)|\}}{\sqrt{2\pi}}\frac{2\a}3 e^{\frac\a 3 |\im e^{-\zeta}|}.
\]
Furthermore, on the boundary,
$\hat \k(e^{-\theta})$ and $\hat \k(e^{-(\theta - \pi i)}) = \hat \k(-e^{-\theta}) = \hat \k(e^{-\theta})$
(as $\k(-t) = \k(t)$)
are rapidly decreasing in $\theta$ and faster than exponentials for $\theta \to - \infty$,
since $\hat \k(p)$ is a Schwartz class function with $\hat \k(0) = 0$,
therefore, $\hat \k(e^{-\theta})e^{-n\theta}$
tends to $0$ rapidly as $\theta \to \pm \infty$, where $n > 0$.

\begin{proposition}
 For $f(\zeta) = \exp\left(i\a e^{-\zeta} - i\beta e^{\zeta}\right)$,
 where one of $\a$ and $\b$ is nonzero. Then
 Then $n_\pm(\mfbar\shift) = \infty$.
\end{proposition}
\begin{proof}
We may assume that $\a > 0$, as the case $\b > 0$ is similar.
Let $\k$ as above, $n \in \NN$ and consider
the function $\xi_{\k,n}(\zeta) = \hat \k(e^{-\zeta}) \exp(-\frac{i\a}2e^{-\zeta}+\frac{i\b}2e^{\zeta}) e^{-(n+\frac12)\zeta}$.
For $\zeta \in \strip$ we have the following estimate:
\begin{align*}
 |f(\zeta)\xi_{\k,n}(\zeta)| &\le \frac{\max\{\k(t)\}}{\sqrt{2\pi}}\frac{2\a}3 e^{\frac\a 3 |\im e^{-\zeta}|}
 \cdot e^{-\frac\a 2 \im e^{-\zeta} + \frac\b 2 \im e^{\zeta}} e^{-(n+\frac12)\re \zeta} \\
 &\le \frac{\max\{\k(t)\}}{\sqrt{2\pi}}\frac{2\a}3  e^{-(n+\frac12)\re\zeta} 
\end{align*}
as $\im e^{-\zeta} > 0$ and $\im e^\zeta < 0$,
and on the boundary we saw that $|\xi_{\k,n}(\zeta)|$ is Schwartz class.
Thus by Lemma \ref{lm:pl-l2}, $f(\zeta)\xi_{\k,n}(\zeta) \in \hardy$.

Let us check the eigenvalue equation.
We have
\begin{align*}
 f(\theta - \pi i) \xi_{\k,n}(\theta - \pi i)
 &= \exp(-i\a e^{-\theta} + i\b e^{\theta}) \hat \k(e^{-\theta}) \exp\left(\frac{i\a}2e^{-\theta}-\frac{i\b}2e^{\theta}\right) e^{-(n+\frac12)(\theta -\pi i)} \\
 &= \hat \k(e^{-\theta}) \exp\left(-\frac{i\a}2e^{-\theta}+\frac{i\b}2e^{\theta}\right) (-1)^n i e^{-(n+\frac12) \theta} \\
 &= (-1)^n i\xi_{\k,n}(\theta).
\end{align*}
Therefore, depending on whether $n$ is even or odd, $\xi_{\k,n}$ is in one of the deficiency subspaces of $\mfbar\shift$
with eigenvalue $i$ or $-i$. They are clearly linearly independent.
\end{proof}

\subsubsection{Generic singular measures}\label{generic}
As we show in Appendix \ref{essential},
any singular inner function $f$ has essential singularities
on the boundary $\RR, \RR - \pi i$ or $\re \zeta \to \pm \infty$.
If the measure is atomic on the boundary, it is not difficult to find vectors in the deficiency indices
similarly as in the case of atoms at infinity.

In a generic case, we are not able to write explicitly the vectors in the deficiency subspaces of $\mfbar\shift$.
Yet, we are able to show that the deficiency indices are $(\infty, \infty)$
again by finding a self-adjoint extension and a sequence of closed symmetric extension
between them.

\begin{proposition}
 Let $f$ be singular inner as above.
 Then $n_\pm(\mfbar \shift) = \infty$.
\end{proposition}
\begin{proof}
As $f$ has no zero in $\strip$, $f$ has analytic roots $f^\frac1{n}$ for arbitrary $n$.
More explicitly, we can write this root as
\[
 (f(\zeta))^\frac1{n} = \exp\left(\frac i n\int \frac{d\mu(s)}{(1+s^2)}\, \frac{1+e^\zeta s}{e^\zeta - s}\right),
\]
and accordingly we have $\mfbar = (M_{{\overline f}^\frac1n})^n$.

Let $n > 2$. As $f^\frac1n \in H^\infty(\SS_{-\infty,0})$ and $f^\frac1n(\theta - \pi i) = \overline{f^\frac1n(\theta)}$,
it holds that
\[
 M_{{\overline f}^\frac1n}^2\shift \subset M_{{\overline f}^\frac1n}\shift M_{f^\frac1n}
 = M_{f^\frac1n}^*\shift M_{f^\frac1n}
\]
and the last expression is manifestly self-adjoint.

We claim that this extension is not trivial.
Indeed, the domain of the latter operator is $M_{f^\frac1n}^{-1}\hardy$.
This is not equal to $\hardy$. In order to see this, pick a vector $\xi \in \hardy$.
As we saw in Corollary \ref{co:pointwise}, the values $\xi(\theta - i\lambda)$
is bounded by $\|\xi\|$ and $\|\shift\xi\|$.
As $|f^\frac1n(\theta)| = |f^\frac1n(\theta - \pi i)| = 1$ on the boundary,
we have $\|\xi\| = \|M_{f^\frac1n}\xi\|$ and $\|\shift\xi\| = \|M_{f^\frac1n(\cdot - \pi i)}\shift \xi\| = \|\shift M_{f^\frac1n}\xi\|$.
On the other hand, $f(\zeta)$ is a nontrivial singular inner function and
there is $\theta - i\lambda \in \strip$ such that $\left|f(\theta - i\lambda)^\frac1n\right| < 1$.
Take $\xi \in \hardy$. If $M_{f^\frac1n}^{-1}$ preserved $\hardy$,
then $M_{f^\frac1n}^{-m}\xi$ would have to be in $\hardy$ for any $m$.
But it is clear that for any $\xi$ there is $m$ such that
$\frac{\xi(\theta - i\lambda)}{f(\theta - i\lambda)^m}$ does not satisfy
the above pointwise bound for vectors in $\hardy$.
In other words, $M_{f^\frac1n}^{-m}\xi \notin \hardy$, $M_{f^\frac1n}^{-1}$ does not preserve $\hardy$
and $M_{f^\frac1n}^*\shift M_{f^\frac1n}$ is a proper self-adjoint extension.

Take an even $n = 2m$. Then by repeating the argument above,
\[
 \mfbar\shift \;\;=\;\; (M_{f^\frac1{2m}}^*)^{2m}\shift \;\;\subset\;\; (M_{f^\frac1{2m}}^*)^m\shift (M_{f^\frac1{2m}})^m \;\;=\;\; M_{f^\frac12}^*\shift M_{f^\frac12}
\]
is $m$ successive proper extensions, therefore, $n_\pm(\mfbar\shift) \ge m$.
As $m$ is arbitrary, they must be infinite.

\end{proof}

\subsection{Outer functions with decay conditions}\label{outer}
Let us consider outer functions. An outer function can be expressed as the exponential of
a Poisson integral of a kernel $\frac{\log\phi(s)}{1+s^2}$ which is $L^1$.
Consequently, $\phi$ cannot have too strong decay as $\re \zeta \to \pm \infty$.

More concretely, the function with fast decay $\phi(s) = \frac1{e^{-s^\a} + e^{s^\a}}$
satisfies this condition if $0 < \a < 1$, but with $1\le \a$ it does not.
Yet the condition that $\frac{\log\phi(s)}{1+s^2}$ should be $L^1$ can be satisfied in
many ways. In addition, $\phi$ can approach to $0$ around finite $s$ so long as
the $L^1$-condition is satisfied. We are not able to treat such a variety of cases in a general
way.

Here we impose a bound on the decay rate of $\phi(s) = |f(\log s)|$.
Under this condition, the operator $\mfbar\shift$ is essentially self-adjoint.

\begin{lemma}\label{lm:outer}
 Let $f$ be a bounded analytic function on $\strip$ 
 and suppose that there are numbers $A, B \ge 0$, $0 \le \a < 1$ such that
 \[
  \left|\frac1{f(\zeta)}\right| \le A\exp\left(Be^{\a|\re \zeta|}\right).
 \]
 Then $f$ is an outer function. Conversely, if $f$ is a bounded outer function defined
 through an $L^1$ function $\frac{\phi(s)}{1+s^2}$ and if
 there are $A, B \ge 0, 0 \le \a < 1$ such that
 $\frac1{\phi(s)} \le Ae^{B\left(|s|^\a + \left|\frac1s\right|^\a\right)},$ then 
 there is $A_1, B_1> 0$ and it holds that $\left|\frac{1}{f(\zeta)}\right| \le A_1\exp\left(B_1e^{\a|\re \zeta|}\right)$
\end{lemma}
\begin{proof}
 Let $f$ be a function which satisfies the estimate above.
 If $f$ had a Blaschke factor, it would have at least one zero and it contradicts the
 assumed estimate (which separates $f$ from $0$), thus $f$ has no Blaschke factor.
 Suppose that $f$ had a singular inner
 factor. If the singular measure $\mu$ corresponding to the factor has non zero measure at
 $0$ or $\infty$, then the factor contains $e^{i\mu(\{0\})e^{-\theta}}e^{-i\mu(\{\infty\})e^{+\theta}}$
 and they decay as $e^{-Be^{|\re \zeta|}}$ when $\re \zeta \to \mp \infty$, respectively,
 which contradicts the assumption ($\a < 1$). If the singular measure were non-zero and had no atom at $0, \infty$,
 it must tend to zero on the boundary as we show in Appendix \ref{essential}, which would again contradicts
 the assumed estimate. Therefore, $f$ cannot have any nontrivial inner part, namely is an outer function.
 
 Conversely, let $f$ be an outer function defined through $\phi$ as above.
 Note that the modulus of $f$ is given through the imaginary part of the integral,
 which is, by putting $e^\zeta = a + ib$,
 \[
    \im\left(\frac1\pi\int_{-\infty}^\infty \frac{ds}{1+s^2}\, \frac{1+e^\zeta s}{e^\zeta - s} \log \phi(s)\right) =
    \frac1\pi\int_{-\infty}^\infty \frac{b ds}{(a - s)^2 + b^2} \log \phi(s).
 \]
 The latter expression can be estimated as follows.
 Note that it follows from the assumption on $\phi$ that
 $|\log \phi(s)| \le B\left(|s|^\a + \left|\frac1s\right|^\a\right) + |\log A|$.
 
 We estimate these three terms separately. Let us first take $B|s|^\a$.
 We may assume that $a \ge 0$ (the case $a < 0$ is analogous).
 Then for the half-line $s < 0$, $s^2 \le (a - s)^2$ and we have
 \[
   \left|\int_{-\infty}^0 \frac{b ds}{(a - s)^2 + b^2} B|s|^\a\right|
   \le \left|\int_{-\infty}^0 \frac{b ds}{s^2 + b^2} B|s|^\a\right|.
 \]
 Next, note that $|a + s|^\a \le |a|^\a + |s|^\a$ if $0 < \a < 1$. Indeed,
 if we put $F(t) = t^\a$ for $t > 0$, then $F'(t) = \a t^{\a-1}$ which is monotonically decreasing.
 Therefore, we have $|s|^\a = F(s) - F(0) > F(a + s) - F(a) = |a + s|^\a - |a|^\a$. 
 We subdivide the other half-line into $[0,a]$ and $(a,\infty)$.
 There, we have the following estimates, respectively:
 \begin{align*}
   \left|\int_0^a \frac{b ds}{(a - s)^2 + b^2} B|s|^\a\right|
   &\le \left|\int_0^a \frac{b ds}{(a - s)^2 + b^2} B|a|^\a\right|
   \le \pi B|a|^\a, \\
   \left|\int_a^\infty \frac{b ds}{(a - s)^2 + b^2} B|s|^\a\right|
   &\le \left|\int_0^\infty \frac{b ds}{s^2 + b^2} B|a + s|^\a\right| \\
   &\le \left|\int_0^\infty \frac{b ds}{s^2 + b^2} (B(|a|^\a+|s|^\a)\right|,
 \end{align*}
 where we used $\int_{-\infty}^\infty \frac{-b ds}{s^2+b^2} = \pi$
 (recall that $b < 0$).
 Furthermore, we put $\int_{-\infty}^\infty \frac{b |s|^\a ds}{s^2+b^2} = \int_{-\infty}^\infty \frac{|b|^\a |s|^\a ds}{s^2 + 1} =: B_2|b|^\a$.
 Altogether, by taking $B_3 = \max\{2\pi B,B B_2\}$, we have
\begin{align*}
 \left|\int_{-\infty}^\infty \frac{b ds}{(a - s)^2 + b^2} B|s|^\a\right| 
 &\le 2\pi B|a|^\a + B B_2|b|^\a\\
 &\le B_3(|a|^\a + |b|^\a)\\
 &\le 2B_3(|a| + |b|)^\a\\
 &\le 4B_3\left|\sqrt{a^2 + b^2}\right|^\a\\
 &= 4B_3(e^{\re \zeta})^\a.
\end{align*}

Next we consider $B\left|\frac1s\right|^\a$.
By successive changes of variables $s = \frac1t, k = t\sqrt{a^2+b^2}$, we get
\begin{align*}
 \int_{-\infty}^\infty \frac b{(a - s)^2 + b^2}\frac{Bds}{|s|^\a} 
 &= \int_{-\infty}^\infty \frac b{(a - \frac1t)^2 + b^2}\frac{B|t|^\a dt}{t^2}\\ 
 &= \int_{-\infty}^\infty \frac{b B|t|^\a }{a^2 t^2 - 2a t + 1 + b^2 t^2}dt\\ 
 &= \frac B{(\sqrt{a^2+b^2})^\a}\int_{-\infty}^\infty \frac{b |k|^\a}{k^2 -\frac{2a}{\sqrt{a^2+b^2}} k + 1} \frac{dk}{\sqrt{a^2+b^2}}. 
\end{align*}
We claim that this integral is bounded for a fixed $0 < \a < 1$.
For $a = 0$, we have
\begin{align*}
 \left|\int_{-\infty}^\infty \frac{b |k|^\a}{k^2 -\frac{2a}{\sqrt{a^2+b^2}} k + 1} \frac{dk}{\sqrt{a^2+b^2}}\right|
 &=  \int_{-\infty}^\infty \frac{|k|^\a}{k^2 + 1}dk = B_2.
\end{align*}
If $a \neq 0$, we can write $b = \b a$ (note that $\b < 0$) and
\begin{align*}
 \left|\int_{-\infty}^\infty \frac{b |k|^\a}{k^2 -\frac{2a}{\sqrt{a^2+b^2}} k + 1} \frac{dk}{\sqrt{a^2+b^2}}\right|
 &=  \int_{-\infty}^\infty \frac{-\b|k|^\a}{k^2 - \frac2{\sqrt{1+\b^2}}+ 1}\frac{dk}{\sqrt{1+\b^2}} \\
 &=  \int_{-\infty}^\infty \frac{-\b|k|^\a}{\left(k - \frac1{\sqrt{1+\b^2}}\right)^2 + \frac{\b^2}{1+\b^2}}\frac{dk}{\sqrt{1+\b^2}} \\
 &\le  \int_{-\infty}^\infty \frac{-\b\left(|k|^\a + \left|\frac1{\sqrt{1+\b^2}}\right|^\a\right)}{k^2 + \frac{\b^2}{1+\b^2}}\frac{dk}{\sqrt{1+\b^2}} \\
 &\le  \left|\frac{\b}{\sqrt{1+\b^2}}\right|^\a \int_{-\infty}^\infty \frac{|s|^\a}{s^2 + 1}ds 
 + \left|\frac{1}{\sqrt{1+\b^2}}\right|^\a \int_{-\infty}^\infty \frac{1}{s^2 + 1}ds \\
 &\le  B_2 + \pi,
\end{align*}
where we changed the variable $k = \frac{\b s}{\sqrt{1+\b^2}}$. Hence we get
\[
 \left|\int_{-\infty}^\infty \frac b{(a - s)^2 +  b^2}\frac{Bds}{|s|^\a}\right|
 \le \frac B{(\sqrt{a^2+ b^2})^\a}(B_2 + \pi) \le B_3 e^{-\a\re \zeta}
\]
Lastly, it is immediate that
\[
  \left|\int_{-\infty}^\infty \frac{\l|\log A|}{(a - s)^2 +  b^2}ds\right|
 \le \pi|\log A|.
\]

 By putting $B_1 = \frac{4B_3}\pi$ and $A_1 = e^{B_1 + |\log A|}$ we finally obtain
 \begin{align*}
  \left|\frac1{f(\zeta)}\right| &\le 
  \exp\left(\left|\frac1\pi\int_{-\infty}^\infty \frac{ b ds}{(a - s)^2 +  b^2} \log \phi(s)\right|\right) \\
  &\le \exp(B_1(e^{\a\re \zeta} + e^{-\a\re \zeta}) + |\log A|) \\
  &\le A_1\exp\left(B_1e^{\a|\re \zeta|}\right),
 \end{align*}
 as desired.
\end{proof}

\begin{lemma}\label{lm:pl}
 Let $h$ be continuous on $\RR + i[-\pi,0],$ analytic in $\RR + i(-\pi, 0)$ and suppose that
 there is $A, B, B_1 > 0, 0 < \a < 1$ such that
 \[
  |h(\zeta)| \le Ae^{Be^{\a|\re \zeta|}},
 \]
 and on the boundary,
 \[
  |h(\theta)|, |h(\theta - \pi i)| \le Ae^{B_1 \re \zeta},
 \]
 and furthermore on the imaginary line $|h(\zeta)| \le A$ when $\zeta \in i[-\pi, 0]$.
 Then it follows that
 \[
  |h(\zeta)| \le Ae^{B_1\re \zeta}
 \]
 on the whole strip.
\end{lemma}
\begin{proof}
 This is actually only a slight variation of Phragm\'en-Lindel\"of principle \cite[Theorem 12.9]{Rudin87},
 which assumes that the function in question is bounded on the boundary.
 
 First we consider the region $\re \zeta > 0$. The product $h(\zeta)e^{-B_1\zeta}$ can be bounded by the 
 same function $Ae^{Be^{\a |\re \zeta|}}$, and on the boundary with $\re \zeta > 0$ we have a better
 estimate: $\left|h(\zeta)e^{-B_1 \zeta}\right| < A$.
 From the last assumption we also have
 $|h(\zeta)e^{-B_1 \zeta}| \le A$ for $\zeta \in i[-\pi, 0]$.
 
 Let $\a < \b < 1$ and $\epsilon > 0$. It is easy to see from the assumptions that the function
 \[
  h(\zeta) e^{-B_1\zeta}\exp\left(-\epsilon\left(e^{\beta \left(\zeta + \frac{\pi i}2\right)} + e^{-\beta\left(\zeta+\frac{\pi i}2\right)}\right)\right)
 \]
 is strongly decreasing when $\re \zeta \to \infty$ (the key is that $\b < 1$, see \cite[Theorem 12.9]{Rudin87}),
 hence especially is bounded on the half strip $\re \zeta > 0$.
 If we consider a large interval of $\re \zeta$, the maximum modulus principle tells that the maximum is taken on the boundary,
 but actually it occurs on the edges of the half strip if the interval is large enough.
 As $\exp\left(-\epsilon\left(e^{\beta \left(\zeta + \frac{\pi i}2\right)} + e^{-\beta\left(\zeta+\frac{\pi i}2\right)}\right)\right) < 1$,
 this implies that
 \[
 h(\zeta) e^{-B_1\zeta}\exp\left(-\epsilon\left(e^{\beta \left(\zeta + \frac{\pi i}2\right)} + e^{-\beta\left(\zeta+\frac{\pi i}2\right)}\right)\right) < A  
 \]
 for $\re \zeta > 0$, but $\epsilon$ is arbitrary, thus we obtain
 $|h(\zeta) e^{-B_1\zeta}| \le A$, which is equivalent to $|h(\zeta)| \le Ae^{B_1 \re \zeta}$.
 
 We can argue similarly for $\re \zeta < 0$  and obtain the desired bound.
\end{proof}
Although we do not use it, this proof can be easily adopted to the case where
the boundary values are bounded by $Ae^{B_1|\re \zeta|}$.

\begin{proposition}\label{pr:outer}
 Let $f$ be a bounded analytic function on $\strip$ with $\overline{f(\theta)} = f(\theta - \pi i)$
 and suppose that there are numbers $A, B \ge 0$, $0 \le \a < 1$ such that
 \[
  \left|\frac1{f(\zeta)}\right| \le Ae^{Be^{\a|\re \zeta|}}.
 \]
 Then $\mfbar\shift$ is essentially self-adjoint.
\end{proposition}
\begin{proof} 
 By the assumption and Lemma \ref{lm:outer},
 $f$ can be represented as a Poisson integral of a kernel $\phi(s) = |f(\log s)|$:
 \[
  f(\zeta) = \exp\left(\frac i\pi\int_{-\infty}^\infty \frac{ds}{1+s^2}\, \frac{1+e^\zeta s}{e^\zeta - s} \log \phi(s) \right).
 \]
 By the symmetry of $f$, the kernel satisfies $\phi(s) = \phi(-s)$.
 Note that, as $\zeta \in \RR + i(-\pi,0)$, $\im e^\zeta < 0$. This should be kept in mind
 when we consider the boundary value $\im e^\zeta \to 0$. We have
 \begin{align*}
  f(\zeta -\pi i) &= \exp\left(\frac i\pi\int_{-\infty}^\infty \frac{ds}{1+s^2}\, \frac{1-e^\zeta s}{-e^\zeta - s} \log\phi(s)\right) \\
  &= \exp\left(-\frac i\pi\int_{-\infty}^\infty \frac{ds}{1+s^2}\, \frac{1+e^\zeta s}{e^\zeta - s} \log\phi(s)\right)
 \end{align*}
 and this time, the relevant boundary value is considered in the sense that $\im e^\zeta > 0$.
 
 As previous cases, let us find a solution to the eigenvalue equation. We claim that
 \[
  g(\zeta) := \exp\left(-\frac i\pi\int_{-\infty}^0 \frac{ds}{1+s^2}\, \frac{1+e^\zeta s}{e^\zeta - s} \log\phi(s)\right).
 \]
 satisfies the eigenvalue equation. First note that, 
 for $s < 0$, the integrand is continuous in $\zeta$ around $\RR$ as $\re e^\zeta > 0$,
 therefore, it is irrelevant whether $\zeta$ approaches to $\RR$
 from above or below. Then, we have
 \begin{align*}
  g(\zeta - \pi i) &= \exp\left(-\frac i\pi\int_{-\infty}^0 \frac{ds}{1+s^2}\, \frac{1-e^\zeta s}{-e^\zeta - s} \log\phi(s)\right) \\
  &= \exp\left(+\frac i\pi\int_0^{\infty} \frac{ds}{1+s^2}\, \frac{1+e^\zeta s}{e^\zeta - s} \log\phi(s)\right),
 \end{align*}
 and here $\im e^\zeta > 0$. If we consider the product of $f$ and $g$, it holds that
 \[
  f(\zeta - \pi i)g(\zeta - \pi i) = \exp\left(-\frac i\pi\int_{-\infty}^0 \frac{ds}{1+s^2}\, \frac{1+e^\zeta s}{e^\zeta - s} \log\phi(s)\right)
  = g(\zeta),
 \]
 because the approach of $\zeta $ in the integral in $[0,\infty)$ coincide, therefore, they cancel each other,
 while for the integral in $(-\infty, 0)$ the direction of approach is irrelevant, and we obtain the equality.
 
 Now, let us assume that there were a nontrivial vector $\xi \in \ker \left(\shift M_f - i\right)$
 (the case with $\ker \left(\shift M_f + i\right)$ is analogous).
 As explained in Section \ref{ideas}, we consider the function $\frac{\xi(\zeta)}{g(\zeta)}e^{\frac\zeta 2}$.
 We have $|g(\theta)| = 1$, as the support of the integral is concentrated on $s < 0$.
 By Lemma \ref{pr:period}, $\frac{\xi(\zeta)}{g(\zeta)}e^{\frac\zeta 2}$ extends to a periodic function, and on the boundary
 $\frac{\xi(\theta)}{g(\theta)}$ is $L^2$, hence $\frac{\xi(\theta)}{g(\theta)}e^{\frac\zeta 2}$ is locally $L^2$.
 By assumption, $f(\zeta)\xi(\zeta) \in \hardy$, while by Lemma \ref{lm:outer}
 it is immediate that there are $A_1, B_1 > 0$ such that
 \[
  \left|\frac{e^{\frac\zeta 2}}{f(\zeta)g(\zeta)}\right| =
  \left|\exp\left(-\frac i\pi\int_0^\infty \frac{ds}{1+s^2}\, \frac{1+e^\zeta s}{e^\zeta - s} \log\phi(s)\right)\exp\left(\frac\zeta 2\right)\right|
  < A_1e^{B_1e^{\a|\re \zeta|}}.
 \]
 
 Let $\eta(\theta)$ be an arbitrary $L^2$-function with the support contained in $[-R,R]$.
 Then it holds that
 \begin{align*}
 \left|\int d\theta\, \overline{\eta(\theta)} \frac{\xi(\theta + \zeta)}{g(\theta + \zeta)}e^{\frac\zeta 2}\right|
  =&\left|\int d\theta\, \overline{\eta(\theta)} \frac{f(\theta + \zeta)\xi(\theta + \zeta)}{f(\theta + \zeta)g(\theta + \zeta)}e^{\frac\zeta 2}\right| \\
  \le & A_1e^{B_1e^{\a(|\re \zeta| + R)}}\left|\int d\theta\, \overline{\eta(\theta)}f(\theta + \zeta)\xi(\theta + \zeta)\right| \\
  \le & A_1e^{B_1e^{\a R}\cdot e^{\a|\re \zeta|}}\cdot \|\eta\|\cdot \|f\xi\|_{\hardy},
 \end{align*}
 and on $\im \zeta = 0$, we have $\left|\int d\theta\, \overline{\eta(\theta)} \frac{\xi(\theta + \zeta)}{g(\theta + \zeta)}e^{\frac{\theta+\zeta}2}\right| \le e^{\frac{R + \re \zeta} 2}\|\eta\|\cdot\|\xi\|$
 since $\|g(\zeta)\| = 1$. Therefore, the analytic function in $\zeta$
 \[
  f_\eta(\zeta) := \int d\theta\, \overline{\eta(\theta)} \frac{\xi(\theta + \zeta)}{g(\theta + \zeta)}e^\frac\zeta 2
 \]
 is $\pi i$-periodic and hence
 satisfies the conditions of Lemma \ref{lm:pl}, and it follows that
 it is bounded on the whole strip by $e^{\frac{R + \re \zeta} 2}\|\eta\|\cdot\|\xi\|$.
 As in Lemma \ref{lm:polynom}, using the periodicity and the bound, there is a function $f_\eta(\frac{\log z}2)$ on $\CC\setminus\{0\}$
 with $\left|f_\eta(\frac{\log z}2)\right| \le |z|^\frac12 e^{\frac{R} 2}\|\eta\|\cdot\|\xi\|$,
 therefore, it must be constantly zero \cite[Theorem 10.26]{Rudin87}.
 Namely,  $f_\eta(\zeta) = 0$ for any $\eta$ supported in $[-R,R]$.
 
 As $R$ is arbitrary, it follows that $\xi = 0$. Namely, $\ker \left(\shift M_f - i\right) = \{0\}$.
 The other eigenspace can be argued similarly.
\end{proof}

The outer functions which comply with this condition are not generic, as $f$ may have zero
on the boundary. Yet, as for the behavior at $\re \zeta \to \pm \infty$, it is not too restrictive,
since e.g.\! $\frac1{e^{-\zeta^\a} + e^{\zeta^\a}}$ with $\a \ge 1$ cannot satisfy the $L^1$-condition,
as mentioned at the beginning of this section.

For $f$ as in Proposition \ref{pr:outer}, $\mfbar\shift$ is essentially self-adjoint,
in other words, $\overline{\mfbar\shift}$ is self-adjoint. Yet in general, the domain of the closure
of an operator is not easy to describe. For $\mfbar\shift$ as above, we have an explicit
description of the closure.

\begin{proposition}
 Let $f \in H^\infty(\SS_{-\pi,0})$ be an outer function.
 Then $\mfbar\shift$ has a self-adjoint extension:
 $\overline{\mfbar\shift} \subset M_{\overline{f_-}}\shift M_{f_-}$, where
 \[
  f_-(\theta) := \exp\left(\frac i\pi\int_{-\infty}^0 \frac{ds}{1+s^2}\, \frac{1+e^\theta s}{e^\theta - s} \log\phi(s)\right).
 \]
 Accordingly, $n_+(\mfbar\shift) = n_-(\mfbar\shift)$.
 
 In particular, if $f$ is as in Proposition \ref{pr:outer},
 then this extension is actually equal to the closure $\overline{\mfbar\shift}$ and
 $n_\pm(\mfbar\shift) = 0$. 
\end{proposition}
\begin{proof}
 The function $f_-$ is well-defined, as the integral is over the negative half-line
 and the integrand is $L^1$. The integral is real, hence $f_-(\theta)$ has modulus $1$.
 Namely, the operator $M_{f_-}$ is unitary.
 Moreover, $f_-$ can be analytically continued to a bounded function on $\strip$ and the boundary
 value at $\RR - \pi i$ is given by
 \[
  f_-(\theta - \pi i) = \exp\left(-\frac i\pi\int_0^\infty \frac{ds}{1+s^2}\, \frac{1+e^\theta s}{e^\theta - s} \log\phi(s)\right),
 \]
 where $\theta$ approaches to the real line from above.
 On this side of the strip, we have
 \[
  \frac{f(\theta - \pi i)}{f_-(\theta - \pi i)}
  = \exp\left(-\frac i\pi\int_{-\infty}^0 \frac{ds}{1+s^2}\, \frac{1+e^\theta s}{e^\theta - s} \log\phi(s)\right)
  = \overline{f_-(\theta)}.
 \]
 Now we have the following inclusion of symmetric operators:
 \[
  \mfbar\shift = M_{\overline{f_-}}M_{f_-(\cdot - \pi i)}\shift \subset M_{\overline{f_-}}\shift M_{f_-},
 \]
 as in Section \ref{zerosoutside}. The last expression is manifestly self-adjoint
 as $M_{f_-}$ is unitary.

 If $f$ satisfies the assumption of Proposition \ref{pr:outer},
 we know that $\mfbar \shift$ is essentially self-adjoint,
 hence the conclusion follows.
\end{proof}

\subsection{A simple essential self-adjointness criterion: perturbation arguments}\label{perturbation}

\begin{proposition}\label{pr:perturbation}
 Assume that there is $r > 0$ such that $\left|\overline{f(\theta)} - r\right| \le r$ for $\theta \in \RR$.
 Then the operator $\mfbar\Delta$ is essentially self-adjoint.
\end{proposition}
\begin{proof}
 We use the W\"ust theorem \cite[Theorem X.14]{RSII}.
 It is obvious that $r\Delta$ is self-adjoint.
 Now, as $\mfbar\Delta = r\Delta + (\mfbar\Delta - r\Delta)$, if we show that
 $\|(\mfbar\Delta - r\Delta)\xi\| \le \|r\Delta \xi\|$, the desired essential self-adjointness follows.
 
 This inequality is a direct consequence of the assumption:
 \begin{align*}
  \|(\mfbar\Delta - r\Delta)\xi\|^2 &=
  \int d\theta\, \left|\left(\overline{f(\theta)} - r\right)\xi(\theta - \pi i)\right|^2 \\
  &\le   r^2\int d\theta\, |\xi(\theta - \pi i)|^2 \\
  &= \|r\Delta\xi\|^2.
 \end{align*}
\end{proof}
As it holds that $\overline{f(\theta)} = f(\theta - \pi i)$, the condition of the Proposition is
equivalent to the existence of $r$ such that $|f(\theta - \pi i) - r| \le r$.

The condition can be rephrased as follows:
There is $r > 0$ such that the complex number $\overline{f(\theta)}$ is in the disk $\{z\in \CC: |z-r| \ge r\}$.
Especially, if there is $\epsilon > 0$ such that
$-\frac\pi{2} + \epsilon < \arg \overline{f(\theta)} < \frac\pi{2} - \epsilon$, this condition is satisfied
(as $f$ is bounded).

\begin{figure}[ht]
    \centering
\begin{tikzpicture}[line cap=round,line join=round,>=triangle 45,x=1.0cm,y=1.0cm]
\clip(-4.54,0.38) rectangle (3.52,7.68);
\draw (-3.2,0.38) -- (-3.2,7.68);
\draw [domain=-4.54:3.82] plot(\x,{(-14.8-0*\x)/-3.7});
\draw (2.34,4.72) node[anchor=north west] {$\mathrm{Re} f$};
\draw (-3.16,7.68) node[anchor=north west] {$\mathrm{Im} f$};
\draw (1.4,3.98) node[anchor=north west] {$2r$};
\draw(-0.54,4) circle (2.66cm);
\draw [shift={(-3.2,4)}] (0,0) --  plot[domain=1.22:1.57,variable=\t]({1*1.82*cos(\t r)+0*1.82*sin(\t r)},{0*1.82*cos(\t r)+1*1.82*sin(\t r)}) -- cycle ;
\draw [shift={(-3.2,4)}] (0,0) --  plot[domain=4.71:5.06,variable=\t]({1*1.8*cos(\t r)+0*1.8*sin(\t r)},{0*1.8*cos(\t r)+1*1.8*sin(\t r)}) -- cycle ;
\draw [domain=-3.2000000000000006:3.8200000000000087] plot(\x,{(--7.94--1.71*\x)/0.62});
\draw [domain=-3.2000000000000006:3.8200000000000087] plot(\x,{(-2.98-1.71*\x)/0.62});
\draw (-3.02,2.26) node[anchor=north west] {$\epsilon$};
\draw (-3,6.46) node[anchor=north west] {$\epsilon$};
\end{tikzpicture}
\caption{$f$ should take the value inside the circle.}
    \label{fig:perturbation}
\end{figure}
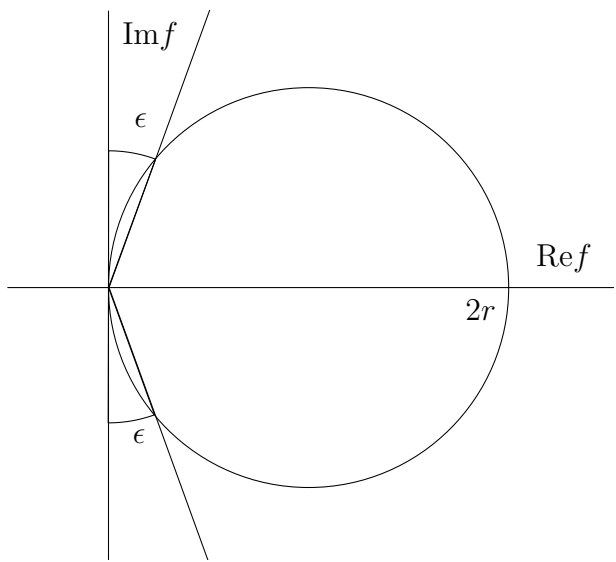

\begin{example}
For $0 < \alpha < 1$, let us consider
\[
 f(\zeta) := \frac1{\cosh\alpha(\zeta+\frac{\pi i} 2)}
 = \frac1{\cosh\alpha\zeta\cdot\cos\frac{\alpha\pi}2 + i\sinh\alpha\zeta\cdot\sin\frac{\alpha\pi}2}.
\]
Note that, as $0 < \alpha < 1$, the denominator is separated from $0$ when
$-\pi < \im \zeta < 0$, and hence $f$ is bounded and analytic.
As $\cosh\frac\alpha{2}\theta > |\sinh\frac\alpha{2}\theta|$ for $\theta \in \RR$,
it holds that $\left|\arg \overline{f(\theta)}\right| < \arg \left(\cos\frac{\alpha\pi}2 + i\sin\frac{\alpha\pi}2\right)$,
which is strictly smaller than $\frac\pi{2}$.
Therefore, we can apply Proposition \ref{pr:perturbation} and $\mfbar\Delta$ is essentially self-adjoint.

Next, for $0 < \beta$, we take
\[
 f(\zeta) := -i\left(\frac1{\zeta - \beta i} - \frac1{\zeta + (\pi +\beta)i}\right)
 = \frac{(\pi+2\beta)}{(\zeta - \beta i)(\zeta + (\pi + \beta)i)}.
\]
It holds that $(\theta - \beta i)(\theta + (\pi + \beta)i) = \theta^2+\beta(\pi + \beta) + \pi \theta i$,
therefore the ratio $|\pi \theta i|/(\theta^2+\beta(\pi + \beta))$ is bounded and
$\arg \overline{f(\theta)}$ is separated from $\frac\pi{2}$, and hence Proposition \ref{pr:perturbation} applies.
\end{example}

On the other hand, it is easy to see that a similar function
\[
 f(\zeta) := \frac1{\zeta - \beta i} + \frac1{\zeta + (\pi +\beta)i}
 = \frac{2\zeta + \pi i}{(\zeta - \beta i)(\zeta + (\pi + \beta)i)}
\]
does not satisfy the condition of Proposition \ref{pr:perturbation}.
Indeed, this function has a zero at $\zeta = -\frac{\pi i}2$ and we know from the results
in Section \ref{finiteblaschke} (see Theorem \ref{th:all} for a more precise argument)
that the operator $\mfbar\Delta$ cannot be essentially self-adjoint.
This non-example tells that the information of zeros of $f(\zeta)$ in the strip
$-\pi < \im \zeta < 0$ which satisfies $\overline{f(\theta)} = f(\theta - \pi i)$
is partially encoded in the behavior of $f(\theta), \theta \in \RR$.
Of course, it follows from the uniqueness of the Blaschke factor, but the relation is rather implicit.
Here, a certain estimate of $\arg f$ can exclude the existence of zeros.

\section{Summary: computing deficiency indices}\label{summary}

\subsection{Composite case}
We have considered three classes of functions $f$ and computed the deficiency indices
of the corresponding operator $\mfbar\shift$: Blaschke products, singular inner functions and
outer functions. A general function $f \in H^\infty(\SS_{-\pi,0})$ is the product of
three factors in $f \in H^\infty(\SS_{-\pi,0})$ as in Section \ref{ideas}:
\[
 f(\zeta) = f_\blaschke(\zeta) f_\tin(\zeta) f_\tout(\zeta),
\]
where we assumed that $c=1$, as a constant does not affect the domain.
Here we assume that $f_\tout$ satisfies the estimate of Proposition \ref{pr:outer}.

First let us suppose that $f_\tin$ is nontrivial.
Then, as we saw in Section \ref{generic}, for an arbitrary $m$ we can take a sequence of proper extensions
\begin{align*}
 \mfbar\shift &\;\;=\;\; \left(M_{f_\tin^\frac1{2m}}^*\right)^{2m}M_{f_\blaschke}^*M_{f_\tout}^*\shift \\
 &\;\;\subset\;\; \left(M_{f_\tin^\frac1{2m}}^*\right)^mM_{f_\blaschke}^*M_{f_\tout}^*\shift \left(M_{f_\tin^\frac1{2m}}\right)^m \\
 &\;\;=\;\; M_{f_\tin^\frac12}^*M_{f_\blaschke}^*M_{f_\tout}^*\shift M_{f_\tin^\frac12}.
\end{align*}
And the expression is still symmetric. This is possible only if $n_\pm(\mfbar\shift) = \infty$.

Next we assume that $f_\tin = 1$.
With the assumption of Proposition \ref{pr:outer},
we saw in Section \ref{outer} that there is $f_- \in H^\infty(\SS_{-\pi,0})$ such that
$\overline{\M_{f_\tout}^*\shift} = M_{f_-}^*\shift M_{f_-}$. Since $M_{f_\blaschke}$ is unitary,
we have
\[
 \overline{\mfbar\shift} = M_{f_\blaschke}^*\overline{M_{f_\tout}^*\shift} = M_{f_\blaschke}^*M_{f_-}^*\shift M_{f_-}
 = M_{f_-}^*M_{f_\blaschke}^*\shift M_{f_-}.
\]
Therefore, the question of deficiency indices is reduced to that of $f_\blaschke$,
as $M_{f_-}$ is unitary. This case has been studied in Sections \ref{finiteblaschke}, \ref{infiniteblaschke}.

Let us summarize the results.
\begin{theorem}\label{th:all}
Let $f \in H^\infty(\SS_{-\pi,0})$ and $\overline{f(\theta)} = f(\theta -\pi i)$ and
\[
 f(\zeta) = f_\blaschke(\zeta) f_\tin(\zeta) f_\tout(\zeta)
\]
be its Beurling factorization. Assume that $f_\tout$ satisfies the estimate of Proposition \ref{pr:outer}.
Then the deficiency indices of the operator $\mfbar\shift$ is
\begin{itemize}
 \item $(\infty, \infty)$ if $f_\tin$ is nontrivial.
 \item $(\infty, \infty)$ if $f_\blaschke$ has infinitely many factors.
 \item $(m,m)$ if $f_\tin$ is trivial and $f_\blaschke$ has $2m$ factors.
 \item $(m+1,m)$ if $f_\tin$ is trivial and $f_\blaschke$ has $2m+1$ factors.
\end{itemize}
\end{theorem}

\subsection{Polar decomposition}
As a byproduct of the analysis above, the polar decomposition of $\mfbar\shift$ can be easily obtained,
again if the outer part of $f$ can be estimated as before.
Let us take the Beurling factorization of $f$:
\[
  f(\zeta) = f_\blaschke(\zeta) f_\tin(\zeta) f_\tout(\zeta),
\]
and assume that $\f_\tout$ satisfies the estimate of Proposition \ref{pr:outer}.
Then we have
\[
\overline{\mfbar\shift} = M_{\overline{f_\blaschke}}M_{\overline{f_\tin}}\cdot \overline{M_{\overline{f_\tout}}\shift},
\]
where $M_{\overline{f_\tout}}\shift$ is essentially self-adjoint by Proposition \ref{pr:outer},
therefore, this is the polar decomposition.

\subsection{von Neumann's criterion}
A von Neumann's criterion \cite[Theorem X.3]{RSII} says that
if a symmetric operator $A$ commutes with an antilinear conjugation $J$,
then $A$ has a self-adjoint extension. Let us see when this can be applied
to $\mfbar\shift$.

Let $(J\xi)(\theta) := \overline{\xi(-\theta)}$ be a conjugation on $L^2(\RR)$.
It is immediate that $J$ preserves $\dom(\shift) = \hardy$.
If we take $f$ such that $\overline{f(\theta)} = f(\theta -\pi i) = f(-\theta)$,
then $\mfbar\shift$ commutes with $J$:
\[
 (J\mfbar\shift\xi)(\theta) = \overline{\overline{f(-\theta)}\xi(-\theta-\pi i)}
 = \overline{f(\theta)}\overline{\xi(-(\theta+\pi i))} = (\mfbar\shift J\xi)(\theta).
\]
As for $f$, a nontrivial singular inner part gives the deficiency indices
$(\infty,\infty)$ and the outer part (with a bound) does not affect the existence
of self-adjoint extensions, as we saw in Theorem \ref{th:all},

Let us take $f(\zeta) = \prod_n \frac{e^\zeta - e^{\a_j}}{e^\zeta - e^{\overline{\a_j}}}$ as a Blaschke product.
From the condition $\overline{f(\theta)} = f(\theta -\pi i) = f(-\theta)$,
it follows that $\a_j$ must appear in a quadruple with $-\pi i-\a_j, -\overline{\a_j}, -\pi i + \overline{\a_j}$.
If $\im \a_j = \frac{\pi}2$, then this is reduced to a pair with $-\overline{\a_j}$.
Furthermore, if $\a_j = -\frac{\pi i}2$, then it does not have a partner but
it must have even multiplicity in order to keep the condition
$\overline{f(\theta)} = f(-\theta)$. Altogether, $f$ has even Blaschke factors and
has equal deficiency indices, in accordance with Proposition \ref{pr:bl-fin}.

\section{Restriction to squares}\label{squares}

As we saw in Section \ref{di}, for a wide class of functions $f$, the operator $\mfbar\shift$ has
nontrivial deficiency indices. The cases where $f$ has zeros, therefore a nontrivial
Blaschke product, is remarkable. If the number of Blaschke factors in $f$ is finite,
we computed explicitly the deficiency indices in Section \ref{finiteblaschke}.
If the singular inner part is trivial, then the deficiency subspaces are finite dimensional.
In particular, if the number of Blaschke factors is odd, then $\mfbar\shift$ does not admits
any self-adjoint extension.
Furthermore, even if $f$ may generically contain a nontrivial singular inner part or the number of Blaschke factors might
be even, the trouble is that there does not seem to be any particular choice of a self-adjoint extension.
Therefore, in order to have a nicer structure, we need to put a constraint on $f$.

Here we consider the case where $f$ is a square of another function: $f(\zeta) = h(\zeta)^2$.
We will see that this choice allows us to find a particular self-adjoint extension
whose spectral calculus can be explicitly performed.
This is not an essential restriction in our application to bound states in quantum field theory
\cite{Tanimoto15-2}. Moreover, this restriction affects only the Blaschke product, because
in the Beurling factorization, the singular inner part and the outer part have no zero,
hence they can be already written as squares.

\subsection{Canonical self-adjoint extensions and operator calculus}\label{canonical}
If $f = h^2$, all the factors in the Beurling factorization (see Section \ref{ideas}) is a square:
\[
 f(\zeta) \;\;=\;\; (h_\blaschke(\zeta))^2\cdot (h_\tin(\zeta))^2\cdot (h_\tout(\zeta))^2 \;\;=\;\; (h_\blaschke(\zeta))^2\cdot (h_\tin(\zeta))^2\cdot f_\tout(\zeta).
\]
Note that $h_\blaschke$ and $h_\tin$ are inner and symmetric, namely,
it holds that
\[
 h_\blaschke(\theta)^{-1} = \overline{h_\blaschke(\theta)} = h_\blaschke(\theta -\pi i),\;\;\;
 h_\tin(\theta)^{-1} = \overline{h_\tin(\theta)} = h_\tin(\theta -\pi i).
\]
Consequently, the multiplication operators $M_{h_\blaschke}, M_{h_\tin}$
are both unitary operators and it holds that $M_{\overline{h_\blaschke}} = M_{h_\blaschke}^*, M_{\overline{h_\tin}} = M_{h_\tin}^*$.

As for $f_\tout$ (we do not need $h_\tout$), let us recall the integral representation, which can be decomposed as follows:
\begin{align*}
 f_\tout(\zeta) &= \exp\left(\frac i\pi\int_{-\infty}^\infty \frac{ds}{1+s^2} \frac{1+e^\zeta s}{e^\zeta - s}\log\phi(s)\right) \\
 &= \exp\left(\frac i\pi\int_{-\infty}^0 \frac{ds}{1+s^2} \frac{1+e^\zeta s}{e^\zeta - s}\log\phi(s)\right)
 \exp\left(\frac i\pi\int_0^\infty \frac{ds}{1+s^2} \frac{1+e^\zeta s}{e^\zeta - s}\log\phi(s)\right) \\
 &=: f_-(\zeta) f_+(\zeta).
\end{align*}
In this decomposition, both $f_+$ and $f_-$ are in $H^\infty(\SS_{-\pi,0})$ as
they are the Poisson integrals of the measures which are the restriction of that for $f$.

Now, $f_-(\zeta)$ is continuous at $\im \zeta = 0$ and $f_-(\theta), \theta \in \RR$, can be
directly computed by the formula above (not only as the boundary value), because
the factor $\frac{1+e^\zeta s}{e^\zeta - s}$ in the integrand is bounded in
the range $s < 0$.
On the other hand, $f_+(\zeta)$ is continuous at $\im \zeta = -\pi$ for the same reason.
Therefore, we have 
\begin{align*}
 f_+(\theta - \pi i) 
 &= \exp\left(\frac i\pi\int_0^\infty \frac{ds}{1+s^2} \frac{1-e^\theta s}{-e^\theta - s}\log\phi(s)\right) \\
 &= \exp\left(-\frac i\pi\int_{-\infty}^0 \frac{ds}{1+s^2} \frac{1 + e^\theta s}{e^\theta - s}\log\phi(s)\right),
\end{align*}
since $\phi(s) = \phi(-s)$. As $\theta \in \RR$, the integrand is real,
hence the exponent is purely imaginary and $f_+(\theta - \pi i)$ has modulus $1$
and it holds that
\[
 f_+(\theta - \pi i) = \overline{f_-(\theta)}.
\]

In general, if $g \in H^\infty(\SS_{-\pi,0})$, not necessarily assuming the symmetry condition $\overline{g(\theta)} = g(\theta - \pi i)$
for $\theta \in \RR$, it holds that
\[
 M_{g(\cdot - \pi i)}\shift \subset \shift M_g
\]
as unbounded operators, namely, for any $\xi \in \hardy$,
$g(\zeta)\xi(\zeta) \in \hardy = \dom(\shift)$
and $g(\theta - \pi i)\xi(\theta - \pi i)$.

Let us consider the case where $f = h^2$.
We have
\[
 \mfbar = M_{\overline{h_\blaschke}}^2M_{\overline{h_\tin}}^2M_{\overline{f_+}}M_{\overline{f_-}}. 
\]
From the observations above, we have the following operator inequality:
\begin{align*}
 \mfbar\shift &\;\;\subset\;\; M_{\overline{h_\blaschke}}M_{\overline{h_\tin}}M_{\overline{f_-}}\shift M_{h_\blaschke}M_{h_\tin}M_{f_-} \\
 &\;\;=\;\; (M_{h_\blaschke}M_{h_\tin}M_{f_-})^*\shift M_{h_\blaschke}M_{h_\tin}M_{f_-},
\end{align*}
where we used the unitarity of $M_{h_\blaschke}, M_{h_\tin}, M_{f_-}$ (separately).
The last expression is manifestly self-adjoint.

As readily seen, this self-adjoint extension of $\mfbar\shift$ has a very particular property.
It is unitarily equivalent to $\shift$, in particular,
its spectrum coincides with that of $\shift$, and for any Borel function $g$
it holds that 
\begin{align*}
 & g((M_{h_\blaschke}M_{h_\tin}M_{f_-})^*\shift M_{h_\blaschke}M_{h_\tin}M_{f_-}) \\
 &= (M_{h_\blaschke}M_{h_\tin}M_{f_-})^*g(\shift) M_{h_\blaschke}M_{h_\tin}M_{f_-},
\end{align*}
in the sense of operator calculus.

Let us summarize this construction.
\begin{theorem}\label{th:ext}
 Let $f = h^2$ and $h\in H^\infty(\SS_{-\pi,0})$.
 Then $\mfbar\shift$ has a self-adjoint extension which is unitarily equivalent to $\shift$,
 where the unitary operator intertwining them is $M_{h_\blaschke}M_{h_\tin}M_{f_-}$ as above.
\end{theorem}

The straightforward treatment of $h_\blaschke, h_\tin$ and that of $f_\tout$ which requires a further factorization
into $f_\pm$ may look discrepant.
Yet, as in Theorem \ref{th:all}, if the outer part satisfies the hypothesis of Proposition \ref{pr:outer},
then this discrepancy does not appear: the extension above takes the form
\[
(M_{\overline{h_\blaschke}}M_{\overline{h_\tin}})^*\cdot \overline{M_{f_\tout}\shift}\cdot M_{h_\blaschke}M_{h_\tin},
\]
where one does not see $f_\pm$.
As we remarked in Section \ref{outer}, the hypothesis is not too restrictive.

\begin{example}
 Let $f$ be a positive constant $f(\zeta) = c > 0$.
 The corresponding $f_+$ and $f_-$ can be easily found: namely,
 $f_-(\zeta) = e^{\frac{i\zeta\log c}\pi}, f_+(\zeta) = ce^{-\frac{i\zeta\log c}\pi}$.
 It is immediate to see that $f_+(\theta)f_-(\theta) = f(\theta) = c$ and
 $f_+(\theta -\pi i) = \overline{f_-(\theta)}$.

 Next, let $f$ be the following outer function $f(\zeta) = \exp\left(-\left(\zeta+\frac{\pi i}2\right)^2\right)$.
 This satisfies $\overline{f(\theta)} = f(\theta - \pi i)$ for $\theta \in \RR$, and indeed outer
 by Proposition \ref{lm:outer}. Let us find out $f_+$ and $f_-$.
 We claim that
 \[
  f_-(\zeta) = \exp\left(-\frac{i}{3\pi}\zeta^3 - \frac{\pi i}{12}\zeta\right),
  f_+(\zeta) = \exp\left(\frac{i}{3\pi}(\zeta+\pi i)^3 + \frac{\pi i}{12}(\zeta+\pi i)\right).
 \]
 It is easy to check that they are again outer and $f_+(\theta - \pi i) = \overline{f_-(\theta)}$ for $\theta \in \RR$.
 By a straightforward computation, one sees $f_+(\theta)f_-(\theta) = f(\theta)$. As $|f_-(\theta)| = 1$,
 this is the desired decomposition. Therefore, we have
 \[
  \mfbar\shift \subset M_{f_-}^*\shift M_{f_-}.
 \]
 Actually, the left-hand side is essentially self-adjoint in this case.
\end{example}

\subsection{A characterization of the canonical extension}
Let us reconsider a Blaschke product $\displaystyle f(\zeta) = \frac{e^\zeta - e^{\a}}{e^\zeta - e^{\overline{\a}}}\frac{e^\zeta - e^{\overline{\a-\pi i}}}{e^\zeta - e^{\a-\pi i}}$.
In Section \ref{finiteblaschke}, we saw that the deficiency indices of $\mfbar\shift$ is
$(1,1)$. Vectors in the deficiency subspaces can be explicitly given as follows:
\[
 \xi_\pm(\zeta) = \frac{e^{\left(1\pm\frac12\right)\zeta}}{(e^\zeta - e^{\a})(e^\zeta - e^{\overline{\a-\pi i}})}
\]
Note that if $\im \a \neq -\frac\pi 2$, then $\zeta_\pm$ have two simple poles.
When $\im \a = -\frac\pi 2$, they have a double pole.
By the standard argument, there are self-adjoint extensions parametrized by $S^1$.
And generically there does not seem to be any particular choice.


The situation is different if $\im\a = -\frac\pi 2$. Then the function $f$ has a double zero at $\z = \a$.
It is easy to observe that the combination $\zeta_+ -e^\a\zeta_-$ has only a simple pole
at the double zero of $f$. It turns out that $\mfbar\shift$ can be extended by including this vector
in the domain.

This example can be generalized in the following form: If $f$ is a square $f=h^2$
and has the outer part which is bounded below,
the above idea allows us to rewrite the domain studied in Section \ref{canonical} from a different point of view,
and prove the self-adjointness in a direct way based on the definition.
\begin{proposition}
 Let $f = h^2$, where $h$ is a bounded analytic function on $\strip$, and assume that the boundary values
 of $h$ satisfy $\overline{h(\theta)} = h(\theta - \pi i)$ and are bounded below.
 Then $\mfbar\shift$ has a self-adjoint extension $A_f$ with the domain
\[
\left\{\xi \in \H:
 h\xi \in \hardy
\right\},
\]
with the action $(A_f\xi)(\theta) = f(\theta -\pi i)\xi(\theta -\pi i)\;$ ($=h(\theta -\pi i)^2\xi(\theta -\pi i)$,
which is well defined on the domain above). This extension $A_f$ coincides with
the one given in Theorem \ref{th:ext}
\end{proposition}
\begin{proof}
 First we need to check that the operator $A_f$ defined above is symmetric.
 This is immediate because if $\xi, \eta$ belong to that domain, then
 \begin{align*}
  \<\eta, \mfbar\shift\xi\> &= \int \overline{\eta(\theta)}f(\theta -\pi i)\xi(\theta -\pi i)d\theta \\
  &= \int \overline{\eta(\theta)h(\theta)} h(\theta- \pi i)\xi(\theta -\pi i)d\theta \\
  &= \int \overline{\eta(\theta-\pi i)h(\theta-\pi i)} h(\theta)\xi(\theta)d\theta \\
  &=  \<\mfbar\shift\eta, \xi\>,
 \end{align*}
 where we can use the Cauchy theorem since $h\xi, h\eta \in \hardy$ (see Proposition \ref{pr:symmetric}).

 Now let $\eta \in \dom(A_f^*)$.
 This means that for any $\xi$ such that $h\xi \in \hardy$,
 we have $|\<\eta, A_f\xi\>| \le C\|\xi\|$  where $C$ is a constant independent of $\xi$.
 In other words,
 \begin{align*}
  C\|\xi\| \ge  &\left|\int \overline{\eta(\theta)}h(\theta-\pi i)^2\xi(\th - \pi i)d\theta\right| \\
  =&\left|\int \overline{\eta(\theta)h(\theta)}h(\theta-\pi i)\xi(\th - \pi i)d\theta\right| \\
  =&\left|\<h\eta, \shift h\xi\>\right|. 
 \end{align*}
 Furthermore, by assumption $h$ is bounded below on the boundary, hence there is $C_h$ such that
 $\|\xi\| \le C_h\|h\xi\|$, and therefore, $\left|\<h\eta, \shift h\xi\>\right| \le CC_h\|h\xi\|$.

 Note that, for any $\check{\xi} \in \hardy = \dom(\shift)$, there is $\xi \in \dom(A_f)$ such
 that $\check{\xi} = h\xi$. Indeed, the boundary value $\frac{\check{\xi}(\theta)}{h(\theta)}$
 is $L^2$ since $h$ is bounded below and $\check{\xi}$ is $L^2$,
 and $\frac{\check{\xi}}h \cdot h = \check{\xi} \in \hardy$, hence
 $\frac{\check{\xi}}h \in \dom(A_f)$ by definition.
 It follows from this that $h\eta \in \dom((\shift)^*) = \dom(\shift) = \hardy$,
 which implies by definition that $\eta \in \dom(A_f)$. This concludes the proof of self-adjointness
 of $A_f$.
 
 To see that this coincides with the extension in Theorem \ref{th:ext},
 note that the assumption that $h$ is bounded below on the boundary implies that the outer factor
 is bounded below on the whole strip and makes no effect on the domain.
 Now, the condition that $h\xi\in\hardy$ is equivalent to $h_\blaschke h_\tin\xi \in\hardy$,
 which is exactly the domain obtained in Theorem \ref{th:ext}.
\end{proof}

The above simple and seemingly natural description of the domain fails if $f$ is not bounded below.
Indeed, if $h$ is an outer function with a simple zero on the boundary,
the closure of $\mfbar\shift = M_{\overline h}^2\shift$
contains functions $\xi$ which have double pole at the boundary.
For such $\xi$, $h\xi \notin \hardy$, therefore, the domain above cannot be the domain of self-adjointness.

\section{Concluding remarks}
For the computation of deficiency indices, the case where the outer part does not satisfy
the assumption of Proposition \ref{pr:outer} remains open.
Yet, we studied this problem with a motivation which arose in Quantum Field Theory (QFT),
and the class of functions which we considered in Section \ref{canonical} seems to suffice
in the operator-algebraic treatment of certain integrable QFT \cite{Tanimoto15-2}.

The operator $\mfbar\shift$ appeared in \cite{CT15-1} as the one-particle component
of the bound state operator. The whole operator is a complicated symmetrization of such a one-particle component
which requires different techniques and we will investigate it in a separate paper \cite{Tanimoto15-2}.


\subsubsection*{Acknowledgements}
I wish to thank Marcel Bischoff, Gian Michele Graf and Berke Topacogullari for interesting discussions.

This work was supported by Grant-in-Aid for JSPS fellows 25-205.

\appendix
\section{Hardy space and Fourier transforms}\label{fourier}
The following is a well-known observation (c.f.\! \cite[Theorem IX.13]{RSII} \cite[Corollary III.2.10]{SW71}),
yet we state it here explicitly, since the sketch of proof in \cite[Problem IX.76(c)]{RSII}
contains a subtle argument, while \cite{SW71} is written for analytic functions of several variables.
Let $M_{e_{\pi}}$ be the multiplication operator on $L^2(\RR, dt)$ by $e_{\pi}(t) = e^{\pi t}$,
which is unbounded. If we take its domain as
\[
 \{\hat g \in L^2(\RR,dt): e^{\pi t}\hat g(t) \mbox{ is again } L^2\},
\]
then it is self-adjoint \cite[Theorem VIII.4]{RSI}.
Consider the dense subspace of compactly supported smooth functions, which is a core of $M_{e_{\pi}}$.
Let $g$ be a function whose Fourier transform $\hat g$ is in that core:
\[
 \hat g(\theta) = \frac1{\sqrt{2\pi}}\int dt\, e^{-it\theta} g(t).
\]
This Fourier transform is defined pointwise, and in this case
$g$ is the inverse Fourier transform of $\hat g$ and is
an entire function of $\theta$ by Morera's theorem \cite[Theorem 10.17]{Rudin87}.
The analytic shift of $g$ by $-i\pi$ is given by
\begin{align*}
 (\shift g)(\theta) = g(\theta-\pi i) &= \frac1{\sqrt{2\pi}}\int dt\, e^{i(t-\pi i)\theta}\hat g(t) \\
 &= \frac1{\sqrt{2\pi}}\int dt\, e^{it\theta}(M_{e_{\pi}}g)(t) 
\end{align*}
or in other words $\widehat{\shift g} = M_{e_\pi}g$.
Therefore, $\shift$ is essentially self-adjoint on the space of functions which are the inverse Fourier transform
of compactly supported smooth functions and unitarily equivalent to $M_{e_{\pi}}$.

For each $g$ such that $\hat g \in \dom(M_{e_{\pi}})$, $g \in \hardy$, because for $\zeta \in \RR+i(-\pi,0)$
the integral
\[
 g(\zeta) = \int dt\, e^{it\zeta}\hat g(t) =  \int_0^{\infty} dt\, e^{-t\epsilon}\cdot e^{it(\zeta - i\epsilon)}\hat g(t) + \int_{-\infty}^0 dt\, e^{it\zeta}\hat g(t)
\]
is $L^1$, where $\epsilon > 0$ is such that $\zeta - i\epsilon \in \RR +i(-\pi,0)$,
which shows the analyticity of $g$ by Morera's theorem, and the uniform $L^2$-bound
follows from $\|M_{e_{i\zeta}}\hat g\|_{L^2} \le \|\hat g\|_{L^2} + \|M_{e_{\pi}}\hat g\|_{L^2}$
and the Plancherel theorem, where $M_{e_{i\zeta}}$ is the multiplication operator by
the function $e^{it\zeta}$.

Now we give an elementary proof of the converse inclusion without using the results of \cite[Chapter III]{SW71}
cited in Section \ref{hardy}.
\begin{proposition}\label{pr:integralrep}
 For each $\xi \in \hardy$, $\hat \xi \in \dom(M_{e_{\pi}})$.
\end{proposition}
\begin{proof}
 Let $g$ be such that $\hat g$ has compact support and is smooth. Then the inverse Fourier transform $g$ is entire analytic,
 $g(t) = \frac1{\sqrt{2\pi}}\int d\theta\, e^{it\theta}\hat g(\theta)$
 and $\bar{g}(\zeta) := \overline{g(\overline \zeta)}$ is entire as well. 

 We claim that
 \[
  \int d\theta\, \xi(\theta + \l i) \overline{g(\theta - \l i)}
 \]
 does not depend on $\l \in (-\pi,0)$.
 Note that $\xi(\zeta)\bar{g}(\zeta)$ is analytic. 
 For $\l_1, \l_2 \in (-\pi,0), \l_1 < \l_2$, let us consider the rectangle surrounded
 by $\im \zeta = \l_1, \im \zeta = \l_2, \re \zeta = R, \re \zeta = -R$
 (see Figure \ref{fig:hardy-domain}).
 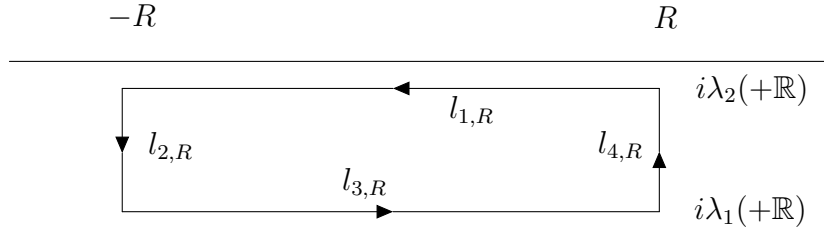
\begin{figure}[ht]
    \centering
\begin{tikzpicture}[line cap=round,line join=round,>=triangle 45,x=1.0cm,y=1.0cm]
\clip(-5.08,-0.7) rectangle (6.52,3);
\draw [domain=-4.3:7.42] plot(\x,{(-1.81-0*\x)/8.24});
\draw [domain=-4.3:7.42] plot(\x,{(--17.84-0*\x)/8.26});
\draw [->] (4.34,1.8) -- (0.8,1.8);
\draw [->] (-2.8,1.8) -- (-2.8,0.94);
\draw [->] (-2.8,0.16) -- (0.8,0.16);
\draw (-2.8,0.94)-- (-2.8,0.16);
\draw (0.8,0.16)-- (4.34,0.16);
\draw [->] (4.34,0.16) -- (4.34,0.96);
\draw (4.34,0.96)-- (4.34,1.8);
\draw (-2.8,1.8)-- (0.8,1.8);
\draw (-3.14,3.06) node[anchor=north west] {$-R$};
\draw (4.12,3.04) node[anchor=north west] {$R$};
\draw (4.68,2.14) node[anchor=north west] {$i\l_2(+\mathbb{R})$};
\draw (4.66,0.52) node[anchor=north west] {$i\l_1(+\mathbb{R})$};
\draw (1.4,1.84) node[anchor=north west] {$l_{1,R}$};
\draw (-2.6,1.36) node[anchor=north west] {$l_{2,R}$};
\draw (-0.02,0.88) node[anchor=north west] {$l_{3,R}$};
\draw (3.38,1.38) node[anchor=north west] {$l_{4,R}$};
\end{tikzpicture}
    \caption{The integral contour for the Cauchy theorem}
    \label{fig:hardy-domain}
\end{figure}
 By the Cauchy theorem,
 the integral of $\xi \hat g$ over this contour is $0$.
 By assumption $\xi(\cdot + i\l)$ is $L^2$ at each fixed $\l \in (-\pi,0)$,
 thus the product $\xi(\cdot + i\l)\overline{g(\cdot - i\l)}$ is $L^1$ at each fixed $\l \in (-\pi,0)$.
 Therefore, as $R$ tends to $\infty$, the integral
 \[
  \int_{-R}^R d\theta\, \xi(\theta + \l i) \overline{g(\theta - \l i)}
 \]
 tends to $\int_{-\infty}^\infty dt\, \xi(\theta + \l i) \overline{g(\theta - \l i)}$.
 What remains to show is the integral on the vertical lines $l_{2,R}, l_{4,R}$ tend to $0$.
 Actually, it is enough to prove that there is a growing sequence $\{R_n\}$
 for which these integrals on tend to $0$, since
 \begin{align*}
  0 &= \lim_{R_n\to \infty} \int_{l_{1,R_n}\cup l_{2,R_n}\cup l_{3,R_n}\cup l_{4,R_n}}d\zeta\, \xi(\zeta)\bar{g}(\zeta) \\
  &= \int_{-\infty}^\infty d\theta\, \left(\xi(\theta + \l_1 i) \bar{g}(\theta + \l_1 i)
  - \xi(\theta + \l_2 i) \bar{g}(\theta + \l_2 i)\right)
   + \lim_{R_n\to \infty} \int_{l_{2,R_n}\cup l_{4,R_n}}d\zeta\, \xi(\zeta)\bar{g}(\zeta).
 \end{align*}
 Now, in order to find such a sequence, note that $\xi(\theta + i\l)\bar{g}(\theta + i\l)$ is an $L^1$ as a function
 with two variables $\theta$ and $\l$, as the $L^2$-norm of $\xi$ is uniformly bounded
 by assumption. This implies that the two-variables integral
 \[
  \int_n^{n+1} d\theta \int_{\l_2}^{\l_1} d\l\, \left|\xi(\theta + i\l)\bar{g}(\theta + i\l) \right|
 \]
 tends to $0$ as $n\to \infty$. Let us say that the integral above is less than $\e_n$,
 where $\e_n \to 0$.
 Then, there must be a set in $[n,n+1)$ with Lebesgue measure larger than $\frac12$ such that
 for $\theta$ in this set we have the one-variable integral
 \[
  \int_{\l_2}^{\l_1} d\l\, \left|\xi(\theta + i\l)\overline{g(\theta - i\l)}\right| < 2\e_n.
 \]
 Similarly, there must be another set in $[n,n+1)$ with measure larger than $\frac12$
 for which it holds
 \[
  \int_{\l_2}^{\l_1} d\l\, \left|\xi(-\theta + i\l)\overline{g(-\theta - i\l)}\right| < 2\e_n.
 \]
 As the interval $[n,n+1)$ has measure $1$, there must be a nontrivial intersection
 of these two sets. We can take $R_n$ from this intersection.
 
 As $\xi \in \hardy$, there is a constant $C_\xi$ such that
 $\|\xi(\cdot + i\l)\|_{L^2} \le C_\xi$ for $\lambda \in (-\pi, 0)$.
 Let us consider $\xi_{-\frac{\pi i}2}(t) := \xi\left(\theta - \frac{\pi i}2\right)$.
 Now, if $\hat g$ has a compact support, then $M_{e_{\l'}}\hat g$ has again a compact support for any $\l' \in \RR$.
 For $\l' \in (-\frac\pi 2,\frac\pi 2)$ , we have
 \begin{align*}
  \< M_{e_{\l'}}\hat g, \widehat{\xi_{-\frac{\pi i} 2}}\> &= \int dt\, \overline{e^{\l' t}\hat g(t)}\widehat{\xi_{-\frac{\pi i} 2}}(t) \\
  &= \int dt\, \overline{g\left(\theta - i\l'\right)}\xi\left(\theta - \frac{\pi i}2\right) \\
  &= \int dt\, \overline{g(\theta)}\xi\left(\theta - i\l' - \frac{\pi i}2\right),
 \end{align*}
 which implies that $|\< M_{e_{\l'}}\hat g, \widehat{\xi_{-\frac{\pi i} 2}}\>| \le \|g\|\cdot \|\xi(\cdot + i\l' - \frac{\pi i}2)\|_{L^2} \le C_\xi\|\hat g\|$.
 Such $g$'s form a core of $M_{e_\l'}$, hence we obtain that $\widehat{\xi_{-\frac{\pi i} 2}}$ is in the domain of $M_{e_\l'}$ and
 $\|M_{e_\l'}\xi_{-\frac{\pi i}2}\| \le C_\xi$. By this uniform bound for $\l' \in (-\frac\pi 2,\frac\pi 2)$,
 it follows that $\widehat{\xi_{-\frac{\pi i} 2}} \in \dom(M_{e_{-\frac\pi 2}})\cap \dom(M_{e_{\frac\pi 2}})$.
 
 It is now clear that we have $\hat \xi = M_{e_{-\frac \pi 2}}\widehat{\xi_{-\frac{\pi i} 2}}$, which belongs to $\dom(M_{e_{\pi}})$.
\end{proof}
Summarizing, $\hardy$ is the inverse Fourier transform of $\dom(M_{e_{\pi}})$.

\begin{corollary}\label{co:pointwise}
 For each $\xi \in \hardy$ and $\l \in (-\pi, 0)$, it holds that
 $\xi(\theta + i\lambda) \le \frac1{\sqrt{-4\pi\lambda}}\|\xi\| + \frac1{\sqrt{4\pi(\pi+\lambda)}}\|\shift\xi\|$.
\end{corollary}
\begin{proof}
 By Proposition \ref{pr:integralrep}, we have the pointwise expression
 \[
  \xi(\theta + i\lambda) = \frac1{\sqrt{2\pi}}\int dt\, e^{it(\theta + i\lambda)}\hat \xi(t),
 \]
 where $\hat \xi$ is an $L^2$-function and in the domain $\dom(M_{e_{\pi}})$.
 In particular, $\hat \xi(t)$ and $\hat \xi(t)e^{\pi t}$ are $L^2$.
 Now, we have
 \begin{align*}
  \xi(\theta + i\lambda)
  &\le \frac1{\sqrt{2\pi}}\int_{-\infty}^\infty dt\, e^{-t\lambda}\left|\hat \xi(t)\right| \\
  &\le \frac1{\sqrt{2\pi}}\left(\int_{-\infty}^0 dt\, e^{-t\lambda}\left|\hat \xi(t)\right|
  + \int_0^{\infty} dt\, e^{-t(\pi+\lambda)}\left|e^{\pi t}\hat \xi(t)\right|\right) \\
  &\le \frac1{\sqrt{2\pi}}\left(\left(\int_{-\infty}^0 dt\, e^{-2t\lambda}\right)^\frac12\left\|\hat \xi\right\|
  + \left(\int_0^{\infty} dt\, e^{-2t(\pi+\lambda)}\right)^\frac12\left\|M_{e_{\pi}}\hat \xi\right\|\right) \\
  &= \frac1{\sqrt{2\pi}}\left(\frac1{\sqrt{-2\lambda}}\left\|\hat \xi\right\|
  + \frac1{\sqrt{2(\pi+\lambda)}}\left\|M_{e_{\pi}}\hat \xi\right\|\right),
 \end{align*}
 which is the desired estimate.
\end{proof}

\section{Essential singularities in the singular part}\label{essential}
Let $\frac{\mu(s)}{1+s^2}$ be a finite singular measure on $\RR$.
We prove that the singular inner function
\[
  f(\zeta) = \exp\left(i\int \frac{d\mu(s)}{1+s^2}\, \frac{1+e^\zeta s}{e^\zeta - s}\right)
\]
has an essential singularity at each point of the boundary $\zeta \in \RR, \RR - \pi i$
where it holds that $\frac{\mu((a - \epsilon, a + \epsilon))}{2\epsilon} \to \infty$
as $\epsilon \to 0$, where $a = e^\zeta \in \RR$.
The set of such points is not empty if $\mu$ is nontrivial,
indeed it is the case for $a$ almost everywhere with respect to $\mu$ \cite[Theorem 7.15]{Rudin87}.

Let us fix such a $\zeta_0$ on the boundary.  
It is known that $f(\zeta)$ has a radial limit almost everywhere (with respect to the Lebesgue measure $ds$)
\cite[Theorem 11.32]{Rudin87}, whose modulus is $1$ \cite[Theorem 17.15]{Rudin87}.
Hence it is easy to see that there is a sequence $\zeta_j \to \zeta_0$ such that
$|f(\zeta_j)| \to 1$.

We show that the radial limit towards $\zeta_0$ is $0$.
The modulus of the function $f$ is determined by the imaginary part
of the integral.
By writing $e^\zeta = a + ib, a,b \in \RR$, we have
\[
 \im\left(\frac{1+e^\zeta s}{e^\zeta - s}\right) = -\frac{b(1+s^2)}{(a-s)^2 + b^2}.
\]
Note that as $\zeta \in \strip$, $b < 0$.
Now we have
\begin{align*}
\int d\mu(s)\frac{-b}{(a-s)^2 + b^2}
 & > \int_{a+b}^{a-b} d\mu(s)\, \frac{-b}{(a-s)^2 + b^2} \\
 & > \int_{a+b}^{a-b} d\mu(s)\, \frac{1}{2|b|} \\
 & = \frac{\mu((a+b,a-b))}{2|b|}
\end{align*}
and the right-hand side tends to $\infty$ as $b \to 0$ if $a$ is in the set explained above.
The modulus of $f$ is obtained for $e^\zeta = a+ib$ by
\[
 |f(\zeta)| = \exp\left(-\int d\mu(s)\, \frac{-b}{(a-s)^2 + b^2}\right),
\]
which tends to $0$ as $b \to 0$ if $a$ is in the set explained above.

Therefore, if $e^{\zeta_0} = a$, $f$ does not have a unique limit towards $\zeta_0$,
so it must be an essential singularity of $f$.

{\small
\def\cprime{$'$}

}


\begin{thebibliography}{10}

\bibitem{AG93}
N.~I. Akhiezer and I.~M. Glazman.
\newblock {\em Theory of linear operators in {H}ilbert space}.
\newblock Dover Publications, Inc., New York, 1993.
\newblock Translated from the Russian and with a preface by Merlynd Nestell,
  Reprint of the 1961 and 1963 translations, Two volumes bound as one.

\bibitem{BSU96}
Y.~M. Berezansky, Z.~G. Sheftel, and G.~F. Us.
\newblock {\em Functional analysis. {V}ol. {II}}, volume~86 of {\em Operator
  Theory: Advances and Applications}.
\newblock Birkh\"auser Verlag, Basel, 1996.
\newblock Translated from the 1990 Russian original by Peter V. Malyshev.

\bibitem{BL04}
Detlev Buchholz and Gandalf Lechner.
\newblock Modular nuclearity and localization.
\newblock {\em Ann. Henri Poincar{\'e}}, 5(6):1065--1080, 2004.

\bibitem{CT15-1}
Daniela Cadamuro and Yoh Tanimoto.
\newblock Wedge-local fields in integrable models with bound states.
\newblock 2015.
\newblock arXiv:1502:01313, to appear in \textit{Comm. Math. Phys.}

\bibitem{DS88}
Nelson Dunford and Jacob~T. Schwartz.
\newblock {\em Linear operators. {P}art {II}}.
\newblock Wiley Classics Library. John Wiley \& Sons, Inc., New York, 1988.
\newblock Spectral theory. Selfadjoint operators in Hilbert space, With the
  assistance of William G. Bade and Robert G. Bartle, Reprint of the 1963
  original, A Wiley-Interscience Publication.

\bibitem{Friedrichs73}
K.~O. Friedrichs.
\newblock {\em Spectral theory of operators in {H}ilbert space}.
\newblock Springer-Verlag, New York-Heidelberg, 1973.
\newblock Applied Mathematical Sciences, Vol. 9.

\bibitem{Goldberg66}
Seymour Goldberg.
\newblock {\em Unbounded linear operators: {T}heory and applications}.
\newblock McGraw-Hill Book Co., New York-Toronto, Ont.-London, 1966.

\bibitem{Haag96}
Rudolf Haag.
\newblock {\em Local quantum physics}.
\newblock Texts and Monographs in Physics. Springer-Verlag, Berlin, second
  edition, 1996.
\newblock Fields, particles, algebras.

\bibitem{Helmberg69}
Gilbert Helmberg.
\newblock {\em Introduction to spectral theory in {H}ilbert space}.
\newblock North-Holland Series in Applied Mathematics and Mechanics, Vol. 6.
  North-Holland Publishing Co., Amsterdam-London; Wiley Interscience Division
  John Wiley \& Sons, Inc., New York, 1969.

\bibitem{IP15}
A.~Ibort and Perez-Pardo J.M.
\newblock On the theory of self-adjoint extensions of symmetric operators and
  its applications to quantum physics.
\newblock 2015.
\newblock arXiv:1502.05229.

\bibitem{Kato76}
Tosio Kato.
\newblock {\em Perturbation theory for linear operators}.
\newblock Springer-Verlag, Berlin-New York, second edition, 1976.
\newblock Grundlehren der Mathematischen Wissenschaften, Band 132.

\bibitem{LW11}
Roberto Longo and Edward Witten.
\newblock An algebraic construction of boundary quantum field theory.
\newblock {\em Comm. Math. Phys.}, 303(1):213--232, 2011.

\bibitem{RSII}
Michael Reed and Barry Simon.
\newblock {\em Methods of modern mathematical physics. {II}. {F}ourier
  analysis, self-adjointness}.
\newblock Academic Press [Harcourt Brace Jovanovich Publishers], New York,
  1975.

\bibitem{RSI}
Michael Reed and Barry Simon.
\newblock {\em Methods of modern mathematical physics. {I}}.
\newblock Academic Press Inc. [Harcourt Brace Jovanovich Publishers], New York,
  second edition, 1980.
\newblock Functional analysis.

\bibitem{RS55}
Frigyes Riesz and B{\'e}la Sz.-Nagy.
\newblock {\em Functional analysis}.
\newblock Frederick Ungar Publishing Co., New York, 1955.
\newblock Translated by Leo F. Boron.

\bibitem{Rudin71}
Walter Rudin.
\newblock {\em Lectures on the edge-of-the-wedge theorem}.
\newblock American Mathematical Society, Providence, R.I., 1971.
\newblock Conference Board of the Mathematical Sciences Regional Conference
  Series in Mathematics, No. 6.

\bibitem{Rudin87}
Walter Rudin.
\newblock {\em Real and complex analysis}.
\newblock McGraw-Hill Book Co., New York, third edition, 1987.

\bibitem{Rudin91}
Walter Rudin.
\newblock {\em Functional analysis}.
\newblock International Series in Pure and Applied Mathematics. McGraw-Hill
  Inc., New York, second edition, 1991.

\bibitem{Schmuedgen12}
Konrad Schm{\"u}dgen.
\newblock {\em Unbounded self-adjoint operators on {H}ilbert space}, volume 265
  of {\em Graduate Texts in Mathematics}.
\newblock Springer, Dordrecht, 2012.

\bibitem{SW71}
Elias~M. Stein and Guido Weiss.
\newblock {\em Introduction to {F}ourier analysis on {E}uclidean spaces}.
\newblock Princeton University Press, Princeton, N.J., 1971.
\newblock Princeton Mathematical Series, No. 32.

\bibitem{TakesakiII}
M.~Takesaki.
\newblock {\em Theory of operator algebras. {II}}, volume 125 of {\em
  Encyclopaedia of Mathematical Sciences}.
\newblock Springer-Verlag, Berlin, 2003.
\newblock Operator Algebras and Non-commutative Geometry, 6.

\bibitem{Tanimoto15-2}
Yoh Tanimoto.
\newblock Bound state operators and locality in integrable quantum field
  theory.
\newblock 2015.

\bibitem{Teschl09}
Gerald Teschl.
\newblock {\em Mathematical methods in quantum mechanics}, volume~99 of {\em
  Graduate Studies in Mathematics}.
\newblock American Mathematical Society, Providence, RI, 2009.
\newblock With applications to Schr{\"o}dinger operators.

\bibitem{vonNeumann55}
John von Neumann.
\newblock {\em Mathematical foundations of quantum mechanics}.
\newblock Princeton University Press, Princeton, 1955.
\newblock Translated by Robert T. Beyer.

\bibitem{Weidmann80}
Joachim Weidmann.
\newblock {\em Linear operators in {H}ilbert spaces}, volume~68 of {\em
  Graduate Texts in Mathematics}.
\newblock Springer-Verlag, New York-Berlin, 1980.
\newblock Translated from the German by Joseph Sz{\"u}cs.

\bibitem{Yosida80}
K{\^o}saku Yosida.
\newblock {\em Functional analysis}, volume 123 of {\em Grundlehren der
  Mathematischen Wissenschaften [Fundamental Principles of Mathematical
  Sciences]}.
\newblock Springer-Verlag, Berlin-New York, sixth edition, 1980.

\end{thebibliography}
\end{document}